\documentclass[10pt, acmsmall]{acmart}

\newif\ifarXiv
\arXivtrue

\usepackage{graphicx}

\usepackage{xspace}
\usepackage{enumitem}
\usepackage{bm}
\usepackage{hhline}
\usepackage{xcolor}
\usepackage{amsfonts}
\usepackage{amsmath}
\usepackage{amsthm}

\newtheorem*{theorem*}{Theorem}

\newif\ifcomm
\ifcomm
\else
\commfalse
\fi
\ifcomm
\newcommand\bilal[1]{\textcolor{red}{Bilal: #1}}
\newcommand\gil[1]{\textcolor{brown}{Gil: #1}}
\newcommand\ran[1]{\textcolor{orange}{Ran: #1}}
\newcommand\sivaram[1]{\textcolor{violet}{SR: #1}}
\definecolor{applegreen}{rgb}{0.55, 0.71, 0.0}
\newcommand\Wenchen[1]{\textcolor{applegreen}{Wenchen: #1}}
\else
\newcommand\bilal[1]{\textcolor{red}{}}
\newcommand\gil[1]{\textcolor{blue}{}}
\newcommand\ran[1]{\textcolor{orange}{}}
\newcommand\sivaram[1]{\textcolor{violet}{}}
\newcommand\Wenchen[1]{\textcolor{red}{}}
\fi

\usepackage{subfig}
\newcommand{\qMAX}{\ensuremath{q\mbox{-MAX}}\xspace}

\newcommand{\eps}{\epsilon}

\newcommand{\floor}[1]{\left\lfloor#1\right\rfloor}
\newcommand{\ceil}[1]{\left\lceil#1\right\rceil}
\newcommand{\parentheses}[1]{\left(#1\right)}
\newcommand{\brackets}[1]{\left[#1\right]}
\newcommand{\set}[1]{\left\{#1\right\}}

\newcommand{\sys}{{\sc SQUID}\xspace}

\newcommand{\eg}{\textit{e.g.}}

\newcommand{\revise}[1]{#1}

\usepackage{thm-restate}
\usepackage{multirow}

\AtBeginDocument{%
  \providecommand\BibTeX{{%
    \normalfont B\kern-0.5em{\scshape i\kern-0.25em b}\kern-0.8em\TeX}}}

\setcopyright{acmlicensed}
\acmJournal{PACMNET}
\acmYear{2024} \acmVolume{2} \acmNumber{CoNEXT3} \acmArticle{19} \acmMonth{9}\acmDOI{10.1145/3676873}

\received{December 2023}
\received[revised]{June 2024}
\received[accepted]{July 2024}

\usepackage{wrapfig}
\begin{document}

\title[Faster Analytics via Sampled Quantile Estimation]{\sys: Faster Analytics via Sampled Quantile Estimation}


\iftrue
\author{Ran Ben Basat}
\authornote{Alphabetical order.}
\email{(r.benbasat@cs.ucl.ac.uk)}
\affiliation{%
  \institution{University College London}
  \country{UK}
}

\author{Gil Einziger}
\ifarXiv
\authornotemark[1]
\fi
\email{(gilein@bgu.ac.il)}
\affiliation{%
  \institution{Ben Gurion University of The Negev}
  \country{Israel}}

\author{Wenchen Han}
\ifarXiv
\authornotemark[1]
\fi
\email{(wenchen.han.22@ucl.ac.uk)}
\affiliation{%
  \institution{University College London}
  \country{UK}}
  
\author{Bilal Tayh}
\ifarXiv
\authornotemark[1]
\fi
\email{(tayh@post.bgu.ac.il)}
\affiliation{%
  \institution{Ben Gurion University of The Negev}
  \country{Israel}}  

\renewcommand{\shortauthors}{Ben Basat et al.}
\fi

\begin{abstract}
Streaming algorithms are fundamental in the analysis of large and online datasets.
%
A key component of many such analytic tasks is \emph{\qMAX}, which finds the largest $q$ values in a number stream. Modern approaches attain a constant runtime by removing small items in bulk and retaining the largest $q$ items at all times. Yet, these approaches are bottlenecked by an expensive quantile calculation. 

This work introduces a quantile-sampling approach \revise{called \emph{\sys}} and shows its benefits in multiple analytic tasks. 
%
%
%
%
Using this approach, we design a novel weighted heavy hitters data structure that is faster and more accurate than the existing alternatives. 
We also show \sys's practicality for improving network-assisted caching systems with a hardware-based cache prototype that uses \sys to implement the cache policy.
The challenge here is that the switch's dataplane does not allow the general computation required to implement many cache policies, while its CPU is orders of magnitude slower. We overcome this issue by passing just \sys's samples to the CPU, thus bridging this gap.

In software implementations, we show that our method is up to 6.6x faster than the state-of-the-art alternatives when using real workloads.  \revise{For switch-based caching, \sys enables a wide spectrum of data-plane-based caching policies and achieves higher hit ratios than the state-of-the-art P4LRU.}



\end{abstract}

\maketitle
\setcounter{page}{1}
\vspace{-1mm}
\section{Introduction}
\label{sec:intro}
\vspace{-1mm}
High-speed stream processing is essential for data management systems such as Google BigQuery~\cite{BigQuery}, Apache Spark~\cite{Spark}, and Amazon Redshift~\cite{Redshift}. 
Stream processing often focuses on the most frequently occurring items known as the \emph{heavy hitters}~\cite{SketchVisor,Cormode:2005:WHW:1061318.1061325,basat2018memento,shahout2023together}. Caching is a special example of the rule where we try to identify the most valuable items according to some score metric that factors in their recent access patterns~\cite{cache1, cache2,cache3}.
Caching is heavily utilized in many systems due to certain localities across workloads. 

%

Traditionally, simple patterns such as \qMAX, which finds the largest $q$ values in a number stream, were addressed using logarithmic data structures such as a heap or a skiplist (e.g., see the Misra-Gries summary~\cite{misra1982finding} or weighted space saving~\cite{SpaceSavingIsTheBest}).
Recently, due to the increased data volumes, the research optimized  toward more speed, sometimes at the cost of other resources such as added memory (e.g.,~\cite{RHHH, Nitro,StatisticalAcceleration, IMSUM,basat2020faster}). 
Another such example is the
\qMAX algorithm~\cite{qMax} that tailors an $O(1)$ data structure for this pattern, which improves various applications asymptotically (and empirically). In a gist, removing the minimal valued item requires logarithmic complexity, but removing items in bulk can be done in a constant per-item time by gradually performing a linear time quantile calculation and removing items below the quantile. Specifically, for a fixed parameter $\gamma$ (e.g., $\gamma=0.25$), they maintain at most $(1+\gamma)q$ items and evict $\gamma q$ items at once. The amortized complexity is a constant as a maintenance \mbox{operation is performed every linear number of steps. }

Our work revisits the \qMAX problem and observes that finding an exact quantile is rather slow. Instead, we propose a novel algorithm that allows us to use approximate quantiles.  Using approximate quantiles allows us to perform the quantile search on samples rather than all the $q$ items, improving the runtime. 
We argue that while the existing \qMAX algorithm~\cite{qMax} is asymptotically optimal, its speed can further be improved by periodically finding an \emph{approximate} quantile that allows us to evict at least $q\cdot \gamma\cdot \eta$ items, for some $\eta\in(0,1)$, with high probability.
Our work leverages this idea and proposes the \sys algorithm that follows the Las Vegas algorithmic paradigm, which guarantees correctness but whose runtime is a random variable. \sys markedly speeds up \qMAX, which translates to higher throughput in multiple \mbox{analytics tasks that utilize the \qMAX structure.}

Next, we focus on the specific task of finding the heavy hitter elements in weighted streams. In this task, we use our techniques in conjunction with a Cuckoo hashing data structure and get an algorithm that is faster and more accurate than the state-of-the-art. 
The state-of-the-art algorithm for weighted heavy hitters is called Sampled MEDian (SMED)~\cite{IMSUM}, 
it periodically finds an exact percentile and then deletes all smaller valued elements. 
Compared with SMED, we have several improvements; first, our data structure allows \emph{implicit} deletions which marks counters as available to allocate without deleting their contents. This way, elements that were allocated with deleted counters are still tracked until their counter is needed, thus improving the accuracy. Second, because of the deletions strategy, our solution does not require a linear pass over the counters, which is SMED's bottleneck, resulting in higher throughput. Finally, unlike SMED, which assumes knowing an upper bound on the stream length to determine the needed number of samples, we design an adaptive algorithm that increases the number of samples over time while restricting the overall error probability to a user-defined parameter $\delta$. Interestingly, we prove that the overall number of samples (corresponding to the amount of work needed by the algorithm) is at most twice that of an optimal algorithm that knows how many times an approximate quantile calculation is needed.
Our algorithms also leverage AVX vectorized instructions\revise{~\cite{AVX}, which operate on vectors element-wise concurrently,} to facilitate a further speedup.  All in all, our new weighted heavy hitters algorithm is currently the fastest in the literature, which makes it attractive for practical deployments.

{
We show that our approach goes beyond measurement. Namely, we show how \sys accelerates the well-known LRFU~\cite{LRFU} cache policy to a greater extent than prior work with a comparable hit ratio. Such an acceleration is significant as LRFU is considered a prominent but computationally intensive cache policy~\cite{ARC}. }
Finally, we demonstrate how \sys can be applied to and benefit the design of real-world systems.
We pick network-assisted caching as a case study. We show that \sys enables a wide range of caching policies (\eg, LRFU and LRU) to operate entirely on a hardware switch, thus accelerating the cache. 
Existing approaches such as \qMAX require flexibility that is only supported by the switch's CPU, although it is orders of magnitude slower than the switch's dataplane. 
Therefore, state-of-the-art in-network caching systems~\cite{switchkv, netcache} rely on the backend (a server that is accessed when the queried item is not cached on the switch) to determine the cache eviction and admission policies. 
However, such an offloading of computation adds a burden to the backend that lowers its ability to answer queries and makes the system slower to react to changes in the access patterns. Further, on each cache miss their backend must determine which item to replace in order to admit the queried element, thus increasing the latency.
In contrast, \sys allows the switch to take control of the admission and eviction policies \mbox{by allowing its CPU to process just the sampled items.}

\revise{Our main contribution is in providing a significant speedup and enabling hardware implementation. 
Another contribution is the novel logical deletion method for heavy hitter, which markedly reduces the error and avoids a linear pass over the array.} 
We evaluate \sys on analytic tasks and demonstrate up to 6.6x speedup. \sys-HH also outperforms state-of-the-arts including Elastic Sketch~\cite{Elastic}. Further, \sys's in-network caching achieves better hit ratios than the state-of-the-arts while enabling hardware switch acceleration.

\section{Background and formulation}

{
This section provides the necessary background to position our work within the existing literature. We define the \qMAX and weighted heavy hitters problems addressed in this work, explain the state-of-the-art for each problem, \mbox{and qualitatively describe how we differ from previous approaches.  }
\ifarXiv
We \mbox{summarize the notations we use in Table~\ref{tbl:notations}.}
\fi
}

\vspace*{-3mm}
\subsection{Problem Formulations}
\vspace*{-.5mm}
Our problems operate on a stream $\mathfrak{S}$, where at each step we append a new tuple of the form $(id,val)$ to $\mathfrak{S}$. The $id$ is an identifier and $val$ is a real number (e.g., request size or queue length), both depending on the specific measurement task.

\looseness=-1
A \qMAX algorithm supports an \emph{update} method that digests a new tuple and a \emph{query} method that lists the $q$ (a pre-determined quantity) IDs with the maximal aggregated value. The aggregation function also depends on the tackled problem (e.g., summation for heavy hitter algorithms).

In the Heavy Hitters problem, we define the \emph{weighted frequency} of an item $x$ (denoted $f_x$) as the sum of all values that appear with $id=x$. We also define the stream's weight $|S|$ as the total frequency of all identifiers in $S$. A heavy hitters algorithm supports an \emph{update} method that digests a new tuple, a \emph{query} method that receives an identifier ($x$), and returns an approximation to $f_x$ ($\hat{f_x}$), s.t. $\left| f_x -\hat{f_x}\right|\le |S|\cdot \epsilon$. Here, $\epsilon$ is a parameter \mbox{given to the heavy hitters algorithm at initialization.}

\vspace{-2.5mm}
\subsection{Background - \qMAX}
{
The work of~\cite{qMax} defines the \qMAX pattern and suggests the \qMAX algorithm. 
Intuitively, \qMAX leverages the fact that maintaining an ordered list of numbers can be done in logarithmic update time, but periodically removing a constant fraction of the smallest items can be done in linear time. Thus, \qMAX increases the space linearly by a performance parameter $\gamma > 0$. That is, instead of storing just $q$ items, \qMAX stores $q(1+\gamma)$ $(id,val)$ tuples. 
The algorithm works in iterations, each starting with $q$ of the memory entries occupied by the largest items, and $q\cdot \gamma$ entries free. After inserting $q\cdot \gamma$ new items, the data structure becomes full and a \emph{maintenance operation} takes place. In such an operation, \qMAX computes the $q$'th largest value using an $O(q)$ time quantile algorithm and deletes all items smaller than it. Notice that since a maintenance operation is only needed once per $\Omega(q)$ updates, the amortized runtime is $O(1)$ per update. The authors further propose a deamortized variant where each update takes constant time in the worst case.
%
}

For a constant $\gamma$, \qMAX operates in $O(1)$ and requires asymptotically optimal $O(q)$ space. It proposes to use the deterministic Median-of-medians algorithm~\cite{BLUM1973448} to achieve a guaranteed bound on the number of operations needed. Their approach requires an additional $O(q)$-sized auxiliary space to store the medians of medians (typically 25\% more space). Alternatively, some faster randomized algorithms, such as Quick Select~\cite{quickSelect}, still require $\Omega(q)$ time. Indeed, the work of~\cite{BLUM1973448} shows that \emph{any} randomized quantile selection algorithm requires $3q/2-O(1)$ comparisons, and therefore advances in faster (exact) quantile computation may not significantly \mbox{speed up the algorithm of~\cite{qMax}.} 

\looseness=-1
Our approach differs from~\cite{qMax} in the nature of the quantity we seek. More specifically, we present a data structure that can solve the problem using an approximate quantile rather than an exact one.
The essence behind this idea is that we can calculate an approximate quantile, with constant probability, in $O(1)$ time by sampling tuples from the current data structure, thus expediting the~\cite{qMax} algorithm. While our approach implies that we perform slightly more maintenance operations (since our quantiles are not exact), their markedly faster execution makes this a beneficial tradeoff.



\subsection{Background - Heavy Hitters}

Collecting exact counts for elements requires allocating an entry for every element in the measurement task. 
However, in some cases, we may not have enough space to monitor all the elements. Further, applications such as event mining and load balancing~\cite{al2010hedera} are often only interested in the heavy hitters -- the largest subset of network flows. Heavy hitter algorithms often maintain data structures with a fixed number of entries, each containing an element's identifier and counter~\cite{SpaceSavings,DIMSUM, frequent1,frequent2,Cormode:2005:WHW:1061318.1061325,frequent4}. 
As an example, the Space-Saving algorithm~\cite{SpaceSavings}  maintains a cache of $\frac{1}{\epsilon}$ entries, each has an ID and a counter, and guarantees that any element $x$ with $f_x\ge\eps \left|S\right|$ has a counter. Notice that we can have at most  $\frac{1}{\epsilon}$ such elements. When a tuple $(id, val)$ from a monitored element $id$ arrives, Space Saving increases its counter by $val$. When a tuple $(x, v)$ from a non-monitored element ($x$) arrives, Space Saving admits it to the cache, and if the cache is full, it deletes the tuple $(x',v')$ with the smallest $v'$. In this case, the counter of $x$ is set to $v+v'$ to account for potential previous appearances of $x$. 

Optimal heavy hitters algorithms require at least $\frac{1}{\epsilon}$ entries to provide an accuracy guarantee of $\epsilon \left|S\right|$. When the update values are limited to $+1$, there are sophisticated data structures to keep the elements ordered~\cite{WCSS, SpaceSavings} in a constant time. For general updates, the works of~\cite{IMSUM,DIMSUM} suggest an asymptotically time-optimal algorithm at the cost of additional memory. Instead of evicting the minimum, their algorithms periodically evict many items at once by finding a quantile. Since the complexity of finding a quantile is linear, we get a constant amortized complexity. The work of~\cite{DIMSUM} also suggests a way to deamortize \mbox{finding the quantile and attaining a worst-case constant update complexity.}

Asymptotically optimal solutions such as~\cite{IMSUM,DIMSUM++} compute exact quantiles, which is slow in practice. 
The SMED algorithm~\cite{IMSUM} uses randomization to speed up the computation and is closest to our approach. 
Unlike~\cite{DIMSUM} that builds over Space Saving, SMED is designed around the Misra-Gries (MG) algorithm~\cite{misra1982finding}. MG is different than Space Saving in how it handles arrivals of unmonitored items. Specifically, if no counter is free, MG subtracts the minimal counter value from all counters, thereby freeing a counter. Tracking the minimal value means that MG works in $O(\log 1/\eps)$ time, while SMED improves this to constant amortized time by subtracting the \emph{sampled} median.  
However, its maintenance still requires a linear time. As they need to iterate over all the counters and delete small counters. From profiling SMED, and conversations with the authors of~\cite{IMSUM}, this linear pass was identified \mbox{as the main remaining bottleneck of the algorithm.}


Our method improves the runtime by allowing \emph{implicit} deletion of entries without a linear scan. Specifically, we use a ``water level'' value and treat entries with values smaller than the water level as logically deleted. When we couple this approach with a Cuckoo hash table, we can guarantee that the table's load factor is always below a given threshold ensuring fast and efficient operation. Additionally, our implicit deletion method also \emph{improves the accuracy}. This is achieved by estimating items that are logically deleted but whose counters are not yet overwritten using their counter value. Further, if such an item arrives, its counter is incremented by the update value, even if it is below the water level.
Our \sys-HH algorithm uses our \sys{} algorithm as a black box to find approximate quantiles and then it updates the water level according to the discovered quantile.   

\section{Algorithms}

This section introduces \sys and \sys-HH. The former expedites the state-of-the-art \qMAX calculation, which translates to faster implementations of diverse network applications such as Priority sampling, Priority-based aggregation, and network-wide heavy hitters. The latter (\sys-HH) focuses on the heavy hitter problem and expedites its runtime compared to the state-of-the-art. 
Intuitively, \sys{} leverages randomization and the concept of faster maintenance operations by only finding a reasonable pivot rather than an exact one to expedite the \qMAX problem. \sys-HH{} utilizes a Cuckoo hash table and periodically uses \sys{} to retain the largest $q$ elements in the table. In addition, it utilizes logical deletion of items that are guaranteed not to be among the largest $q$ with a sea-level technique that keeps the load factor of the Cuckoo table bounded. These improvements result in \sys{}-HH being considerably faster than the state-of-the-art approaches.

\subsection{\sys: Speeding up \qMAX }
Here, we analyze the bottlenecks of \qMAX{}~\cite{qMax} and discover that the deterministic exact quantile calculation algorithm used in the previous work is needlessly slow. Our approach replaces the exact quantile calculation with a faster approximate quantile calculation and the {deterministic approach with a (faster) randomized one. }
\sys{} and \qMAX{} use an array with $q(1+\gamma)$ elements where at each step, we add a new entry number to the array.  Here, $\gamma\in(0,1)$ is a parameter that presents a space-to-speed tradeoff; the larger $\gamma$ is, the faster the algorithms, but the more space needed.
We store the previously discovered top-$q$ numbers to the left of the array and then add newly arriving items to the right side of the array. Each such addition updates a running index that always points to the next open spot in the array. Eventually, the running index reaches the end of the array; in that case, we need to perform a maintenance operation.  \mbox{This is where \sys{} diverges from \qMAX{}.}

\qMAX{} performs an exact quantile calculation, moves the $q$ largest items to the left of the array, and starts inserting new items from location $q+1$.  Correspondingly, \sys{}'s maintenance performs an \emph{approximate} quantile calculation and shifts items larger than the approximate quantile to the left. The approximation is correct if there are at least $q$ items bigger than it; otherwise, we must repeat the process. At the end of a \sys{} iteration, the running index moves to a location possibly larger than $q+1$. We can find that \sys{} performs more frequent maintenance operations than \qMAX{}, but each of these operations requires less computation.  

\begin{figure}[t!]
\vspace*{-0mm}
\centering
\subfloat[\qMAX{} finds exact quantile.]{\includegraphics[width=0.46796395\columnwidth]{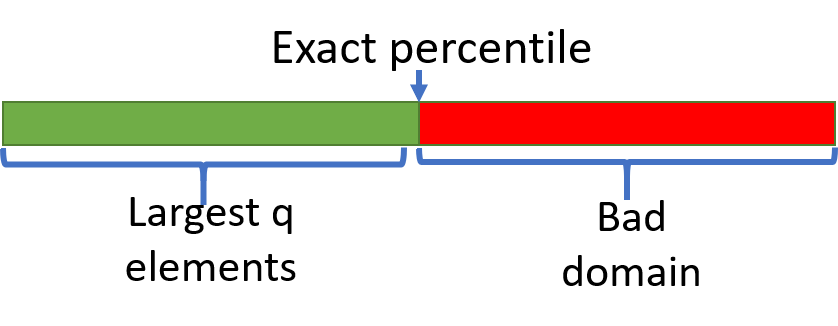}}
\vspace*{-0mm}
\centering
\subfloat[{\scriptsize \sys{} finds approx. quantiles.}]{\includegraphics[width=0.46796395\columnwidth]{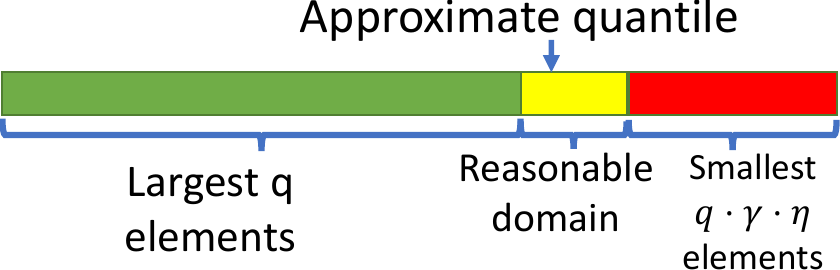}}
\vspace*{-1mm}
\caption{A qualitative comparison of \sys{}'s and \qMAX{}: \qMAX{} searches for the exact quantile, which is time-consuming, while \sys{} settles on any reasonable quantile (in the yellow region). Such quantiles are much easier to obtain, resulting in shorter and more \mbox{frequent maintenance operations.}
}
\vspace*{-2mm}
\label{fig:example}
\end{figure}

\textbf{Finding a good pivot: } Our goal is to find a \emph{pivot} that is smaller than $q$ numbers, and is larger than at least $q\cdot \gamma\cdot \eta$ numbers, for some $\eta\in(0,1)$.
Thus, the pivot is correct because it retains the top $q$ numbers, and it is sufficiently useful as it allows us to at least eliminate $q\cdot \gamma\cdot \eta$ numbers from the array, as they are smaller than the pivot and thus are not in the $q$-largest.
We say that our pivot selection is \emph{successful} if it meets the above conditions.

Intuitively, if one sample $Z$ numbers out of the $q(1+\gamma)$ elements, and we select the $k$'th smallest number in the sample as a pivot, then in \emph{expectation} our pivot will be larger than $q(1+\gamma)\cdot (k/Z)$ of the numbers in the array. 
However,  the actual rank of the pivot is a random variable and may deviate from its expected value. If the sampled pivot is too small, we may clear fewer than $q\cdot \gamma\cdot \eta$ elements and would require another maintenance soon thereafter. On the other hand, if the pivot is among the largest $q$, we must retry the sampling procedure, costing us additional computation.
\sys{} uses an additional parameter $\alpha\in(0,1)$ to aim for the pivot to be larger than $q\cdot \gamma\cdot \alpha$ elements in expectation, by setting  $Z=k\cdot \frac{1+\gamma}{\gamma\cdot\alpha}$.
%
We can then use concentration bounds, as shown in our analysis below, to guarantee with probability $e^{-\Omega\parentheses{Z(1-\alpha)^2}}$ that the pivot would be smaller than $q$ elements.
This suggests a trade-off -- a smaller $\alpha$ implies a higher chance of a successful pivot but also less benefit from the pivot. Similarly, a higher $Z$ value means sampling more elements and spending more time doing maintenance but having a better chance of a successful pivot selection.

To balance these tradeoffs, we introduce an additional parameter: $\delta\in(0,1)$. 
The goal of the sampling is to produce a pivot that, with probability $1-\delta$, satisfies two properties:
\begin{enumerate}[topsep=2pt, partopsep=2pt]

    \item The pivot is \emph{smaller} than {at least} $q$ elements to ensure we track the largest numbers.\label{cond1}
    \item The pivot is \emph{larger} than {at least} $q\cdot\gamma\cdot\eta$ elements,  allowing eviction for enough  numbers.\label{cond2}
\end{enumerate}
Intuitively, when $\delta$ is large, we may get faster execution on average but with a larger variance as the pivot selection is more likely to fail.
This means that with the desired probability, we can guarantee that the maintenance operation would delete $\Omega(q)$ elements to allow sufficiently long iterations. Note that $\delta$ does not need to be very small, as we can verify that the pivot is valid and sample another pivot otherwise. 

We state the following theorem, which quantifies the number of samples $Z$ needed by \sys's maintenance. For simplicity, we avoid rounding $k$ and $Z$ to integers as this adds a lower-order error.
Note that in the following $\eta$ is a function of $\alpha$ that {satisfies $\eta\ge 1/4\cdot  (3 \alpha^2 + 2\alpha - 1)\ge 2\alpha-1$. }
\begin{restatable}[]{thm}{mainthm}
\label{thm:mainthm}
Let $\alpha\in(1/2,1),\delta,\gamma\in(0,1)$ and denote $\eta = 
1/4 \cdot \parentheses{(\alpha+1)^2 +(\alpha-1)\sqrt{\alpha^2 + 14 \alpha + 1}}
$.
Consider a set of numbers $S$ with $|S|=q(1+\gamma)$ elements and denote $k=\frac{2\alpha\ln(2/\delta)}{(1-\alpha)^2}$ and $Z=k\cdot \frac{1+\gamma}{\gamma\cdot\alpha}$.
Then with probability $1-\delta$, the $k$'th smallest number among a uniform sample of $Z$ values from $S$ is larger than at least $q\cdot \gamma\cdot \eta$ numbers \mbox{and smaller than at least $q$ numbers.}
\end{restatable}
\begin{proof}
%
%
%

Let $C$ denote the number of samples smaller than $S_{(q\gamma)}$.
Then the failure event where the sampled value is smaller than less than $q$ numbers in $S$ is equivalent to $C\le k$.
Notice that $C\sim\text{Bin}(Z, \frac{\gamma}{1+\gamma})$ and thus $\mathbb E[C] = \frac{Z\gamma}{1+\gamma}=k/\alpha$.
Therefore, by setting $\Delta=(1-\alpha)$ we have:
{\small\hspace{-5mm}
\begin{align}
\Pr[C \le  k] = \Pr[C \le  \mathbb E[C] (1{-}\Delta)] \le e^{{-}\frac{\Delta^2\mathbb E[C]}{2}}
= e^{{-}\frac{(1{-}\alpha)^2k/\alpha}{2}} = \delta/2
.\label{eq:failureProb}
\end{align}
}
Next, let $D$ denote the number of samples smaller than $S_{(q\gamma\eta)}$.
Then the failure event where the sampled value is larger than less than $q\gamma\eta$ numbers in $S$ is equivalent to $D\ge k$.
Notice that $D\sim\text{Bin}(Z, \frac{\gamma\eta}{1+\gamma})$ and thus $\mathbb E[D] = \frac{Z\gamma\eta}{1+\gamma}=k\eta/\alpha$.
Therefore, by setting $\Delta=(\alpha/\eta-1)$ we have:
$
\Pr[D\ge k] = \Pr[D \ge \mathbb E[D](1+\Delta)] \le  e^{-\frac{\Delta^2\mathbb E[C]}{2+\Delta}} 
= e^{-\frac{(\alpha/\eta-1)^2 k\eta/\alpha}{2+(\alpha/\eta-1)}} = e^{-\frac{(1-\alpha)^2k/\alpha}{2}} = \delta/2$
.

%
Together with~\eqref{eq:failureProb}, this concludes the proof.
%
\end{proof}

As a numeric example, consider the case where $q=10^6,\gamma=1$, which is used in~\cite{qMax}. If we set $\delta=10\%,\alpha=4/5$, we get that we only need to sample $Z=300$ elements (and pick the $k$'th smallest, for $k=120$) to satisfy the requirements. Finding a quantile within $Z=300$ elements is considerably faster, and more space efficient than finding it within $q=10^6$ elements  as the \qMAX algorithm of~\cite{qMax} does. Indeed, the sampled pivot can fail (with probability bounded by $\delta=10\%$), but we can then repeat the sampling procedure. Additionally, we are not guaranteed to clear $q\gamma$ items but just $q\gamma\eta$, for $\eta\approx 0.63$, (and on average, clear $q\gamma\alpha$) so we require more maintenance operations than \qMAX. Yet, the benefit of the faster pivot selection dwarfs these drawbacks.
Interestingly, for finding the $k$'th smallest number in the sample, we could use the \qMAX algorithm of~\cite{qMax}. However, for such a small $k$ a simple heap is fast enough. As another bonus, we only need an additional memory of $k$ elements for the heap, while \qMAX requires $\approx 0.25q(1+\gamma)$ auxiliary space for the pivot calculation that used median-of-medians.

\paragraph{\textbf{A Las Vegas Algorithm}}
As described previously, our algorithm has a chance of failing and we need to keep selecting pivots until we find a successful pivot. That is, we describe a Las Vegas algorithm where the runtime is probabilistic but the result is always correct.  More so, we need to check each pivot selection and verify that it is successful. To do that, we make a linear pass over the array, placing larger than the pivot items on one side and smaller ones on the other. While this requires $\approx q$ comparisons, such a step is significantly faster than exact quantile computation. In addition, it has the advantage of allowing the algorithm to insert elements sequentially into a contiguous memory block during an iteration. Such an access pattern minimizes cache misses and compensates for the additional computation. If the number of elements smaller than the pivot is not in $[q\gamma\eta, q\gamma]$, we select another pivot and start again. 
Theorem~\ref{thm:mainthm} implies that this step is rare, and we experimentally show it barely decreases the algorithm's throughput on real traces.

\paragraph{\textbf{Parameter Tuning}}
Notice that while $\gamma$ is inherent to the problem definition, we have some freedom in choosing $\alpha$ (which would, in turn, determine $\eta$).
While a valid approach is to tune it experimentally, an interesting question from a theoretical perspective is how to decide on the ``right'' $\alpha$ value. 
Intuitively, the above sampling procedure succeeds with probability $1-\delta$, making the number of failed attempts to be a geometric random variable $X\sim\mathit{Geo(1-\delta)}$.
Therefore, the expected runtime of the maintenance procedure is $O(aX+b)$, where $a=O(Z)$ \footnote{It is possible to get the $k$'th smallest element among the $Z$ samples in $O(Z)$ time, e.g., by using a \qMAX data structure.} is the time it takes for the sampling procedure and $b=O(q)$ is \mbox{the time for the pivot verification step.}

Our analysis in 
Appendix~\ref{app:param} 
provides a closed-form formula for the expected \emph{number of comparisons} needed with respect to $\alpha$. Further, the expression is concave, and we can numerically approximate the optimal value to the desired precision. 
Nonetheless, our experiments show that minimizing the number of comparisons is not a reliable proxy for the runtime, as it is affected by other factors such as random bits generation. {Thus, in the evaluation section, we tune it experimentally.}

\paragraph{\textbf{Deamortization}}
We show how \sys can reach $O(1)$ \emph{worst-case} operations in Appendix~\ref{app:deamortization}.

\subsection{Faster Heavy Hitters with \sys-HH}

In this section, we present our  \sys-HH that solves the weighted heavy hitters problem with a constant update time (similar to recent works~\cite{IMSUM,DIMSUM}) but is empirically faster and more accurate. 
Among the existing constant update time solutions, Anderson et al.~\cite{IMSUM} presented the algorithm closest to our approach.
Their solution, \emph{SMED}, stores $C /\epsilon$ counters for some $C>2$. Such a choice is asymptotically optimal but the constant is far from ideal as we can solve the problem with $1 /\epsilon$ counters~\cite{SpaceSavings}. 
When SMED allocates a new entry to each unmonitored flow it encounters. When the number of monitored items becomes greater than $C /\epsilon$, SMED samples $O(\log n)$ counters and calculates their median value. Then, it decreases all counters by this median value, deleting any counter that becomes zero or lower.

\sys{}-HH follows the same template while addressing important drawbacks of SMED. 
First, the number of counters ($C/\epsilon$) for $C>2$ can be reduced toward the optimal $1/\epsilon$ counters. 
Second, as confirmed by both the authors and our experiment, SMED's performance is bottlenecked by the linear-time maintenance operation of updating all counters. Furthermore, SMED requires auxiliary data structures that have their own memory overheads. Namely, it also includes a hash table that maps elements to counters to increase in constant time whenever \mbox{a flow with an allocated counter arrives.}

In \sys-HH we circumvent all of these problems. First, we require no linear-sized auxiliary data structures and only make use of a single Cuckoo hash table~\cite{Cuckoo}. Doing so reduces the memory overheads, and simplifies our algorithms. As for auxiliary structures, we need them to take a sample but such an approach allows for fewer counters (e.g., $1.5/\epsilon$).

\looseness=-1
Finally, and most importantly we do not make a linear pass to remove small entries from our Cuckoo table (avoiding the bottleneck in the previous approach). Instead, we maintain a ``water level'' quantity and treat all table entries below the water level as logically deleted. Thus, by carefully setting the water level we can make sure that our Cuckoo table never fills up, even if we do not know exactly which entries are deleted. 

To explain our water level technique, we need to first provide a brief explanation of how Cuckoo Hash (CH) tables work. 
CH tables are key-value stores that support operations such as mapping keys into values and supporting updates.
For $d,w\in\mathbb N$, CH organizes the (key, value) data into $d$ \emph{buckets}, each containing $w$ pairs.
It uses two hash functions $h_1,h_2:\mathcal K\to \set{0,\ldots,d-1}$ that maps keys $k\in\mathcal K$ into buckets, such that a key $k$ can only be found in buckets $h_1(k)$ or $h_2(k)$.
When inserting a $(k,v)$ pair, CH inserts the pair into either bucket $h_1(k)$ or $h_2(k)$ if not both are full.
In case both buckets are full (contain $w$ pairs each), CH \emph{evicts} a pair $(k',v')$ from one of the buckets, to make room for $(k,v)$.
Next, $(k',v')$ needs to be inserted; if it was evicted from $h_1(k')$, CH inserts it to $h_2(k')$ and vice versa.
Note that this insertion may cause eviction of an additional pair, etc. The operation terminates \mbox{once the inserted pair is placed in a non-full bucket.}

The analysis of CH shows that as long as the \emph{load factor} (defined as the number of pairs in the table divided by $w\cdot d$) is not ``too high'', insertions are likely to terminate after only a few evictions.
The notion of ``too high'' depends on $w$. For $w=2$, a load factor of up to  $1/2$ is likely to be okay, while for $w=4$, load factors of $L\approx 90\%$ are fine~\cite{mitzenmacher2009some}.
Since the lookup of a bucket takes $O(w)$ time, it is common to restrict the bucket sizes, e.g., to $w=4$ or $w=8$. 

Thus, given a CH table, our goal is to never exceed the maximal load factor for correctness, and thus we need to set the water level in a manner that there are always enough logically deleted entries. On the other hand, for correctness, we must retain the largest  $(1/\eps)$ counters at all times. Any counter which is not part of the $(1/\eps)$ largest can be logically deleted safely. 
While operating the CH table, we use logically deleted entries on read operations. However, when updating the table we are free to write over these entries. Such an approach improves the accuracy (as we retain the maximal number of counters at all times), and the speed (as we avoid a periodic linear pass over all entries).  {We now explain how \sys-HH operates:}

\begin{figure}[ht!]
\vspace*{-1mm}
\centering
\includegraphics[width=.62594\columnwidth]{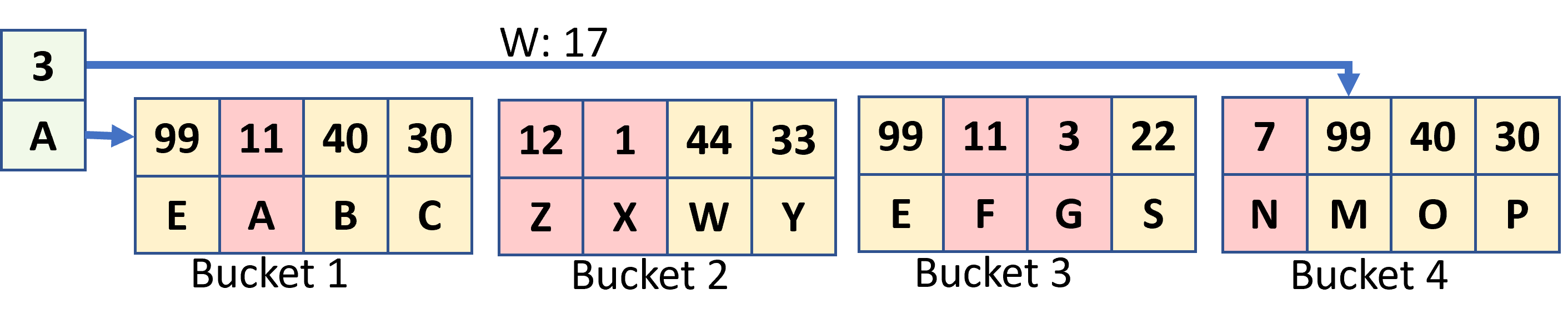}
\vspace{-2.5mm}
\caption{An illustration of \sys-HH. We use the water level technique to keep the load factor of the Cuckoo table within acceptable margins. In this example, the entries below the water level are pink and can be replaced. However, notice that all the tables' entries (even those logically deleted) contain valuable information. In this example, we report another occurrence of flow A with weight 3. According to \sys-HH, if the flow does not have a counter, it is added with a frequency of $\mathcal{W}+3 = 17+3 =20$. In this example, we are lucky to find a logically deleted entry of $A$ with a frequency of $11$, which we increase to $14$. Thus, we would have a smaller approximation error for A's frequency.  
}
\label{fig:squid}
\vspace{-2mm}
\end{figure}

\noindent \textbf{Query: }
The use of CH tables allows finding a flow's counter in a constant time with a lookup. \\
\textbf{Update: }
When seeing a new packet $(id,val)$, we check whether $id$ has a counter in either bucket $h_1(id)$ or $h_2(id)$.
If such a counter exists, we increase its value by $val$. Otherwise, we insert the counter $\langle id, val+\mathcal W\rangle$ into one of the buckets and possibly move an existing counter to its other bucket to make room.
Here, $\mathcal W$ is the \emph{water level} of the algorithm -- a quantity that lower bounds the size of the $(1/\eps)$'th smallest counter in the CH table. Once the number of counters in the CH table reaches $w\cdot d/L(w)$, where $L(w)$ is the allowed load for CH with $w$-sized buckets (e.g., $L(2)=1/2$ and $L(4)\approx 0.9$), a maintenance operation is performed.
Specifically, we use our \sys sampling procedure to produce a value that is (with a very high probability) smaller than at least $1/\eps$ counters and larger than $\alpha/\epsilon$ others. We use this quantity as the water level $\mathcal W$, thereby implicitly deleting any counter below this level. That is, when we search a bucket for an unallocated counter, we consider all counters whose value is lower than $\mathcal W$ to be unallocated. This way, we avoid the expensive linear scan that is the \mbox{bottleneck of heavy hitters algorithms like SMED.}

Finally, we determine the CH dimensions: we use a constant value for $w$ (e.g., $w=4$) to ensure that bucket lookups require constant time. To ensure that the CH works, we set $d=\frac{c}{w\cdot L(w)\cdot \eps} = O(1/\eps)$; here, $c=O(1)$ is a performance parameter with a similar purpose as $\gamma$ in \sys. Specifically, our sampling procedure uses $\gamma = \frac{c}{L(w)}-1$ to determine its parameters. Notice that, unlike \sys, here we can have $c=1$, as we already have a $1/L(w)$ multiplicative constant factor increase in the number of counters. 
Critically, we note that this allows \sys-HH to be not only faster but also use fewer counters than SMED: SMED requires $C\cdot\epsilon^{-1}$ counters for some $C>2$, while \sys-HH can operate using $\frac{\epsilon^{-1}}{L(w)}$. For example, with $w=4$ we can use fewer than $1.15\epsilon^{-1}$ counters.
Interestingly, such a choice of $w$ is also well aligned with the capabilities of modern AVX operations. Specifically, we use AVX to parallelize the checking of whether a key is in a bucket. Similarly, if an unmonitored flow arrives, we can parallelize the search for a counter smaller than the water level in its buckets using AVX. Compared with scalar computation, we measured our AVX code (using AVX2, which is now standard in most CPUs) to further accelerate the computation by $\approx 30-40\%$.
Figure~\ref{fig:squid} provides a detailed example of \sys-HH's update process and of the water level technique. Observe that the Cuckoo hash table is the only data structure we need and that all entries in the hash table (even the logically \mbox{deleted ones) contain useful measurement information.} 

\paragraph{\textbf{A Monte-Carlo HH Algorithm:}}
Heretofore, we referred to \sys as a Las Vegas algorithm that is guaranteed to succeed. However,  it is often acceptable for the weighted HH application to allow a small failure probability as this enables faster algorithms, as done, e.g., in SMED and sketches such as Count Min~\cite{CountMinSketch} and Count Sketch~\cite{topk2}.
Specifically, SMED assumes that we know an upper bound on the number of packets in the measurement $n$ and, therefore, can bound the number of maintenance operations. Knowing the number of packets allows SMED to determine the number of samples in each maintenance so that the overall error probability will be bounded by $1/\text{poly}(n)$.
Namely, it follows that there can be at most $\overline M=O(n\epsilon)$ maintenance operations. Therefore, if each has a failure probability of, e.g., $\delta'=\delta/\overline M$, the stream will be processed without an error with probability $1-\delta$. Notice that getting a failure probability of $\delta'$ requires $\mathfrak c_{SMED}\cdot \log\delta'^{-1}=\mathfrak c_{SMED}\cdot \log(\overline M/\delta)$ samples for some constant $\mathfrak c_{SMED}$ that depends on the redundancy parameter $C$. Thus,  SMED requires \mbox{$\mathfrak c_{SMED}\cdot  M\cdot \log(\overline M/\delta)$ samples in total.}

In \sys-HH we can take a similar approach and have $\mathfrak c_{\sys}\cdot  M\cdot \log(\overline M/\delta)$ samples in total for a constant $\mathfrak c_{\sys}$ that depends on $L(w)$ and $c$ (see~above).\footnote{To simplify the algebra, we use base-2 log and the difference from the $\ln(2/\delta)$ in our analysis can be factored into $\mathfrak c_{\sys}$.} 
However, in some cases, we may not know the stream length beforehand; further, the data skew means that many of the arriving keys have a counter and do not bring us closer to a maintenance stage.
Therefore, it is important to design a solution when \emph{the actual number of maintenances} $M$ is unknown.

Our algorithm works in \emph{phases}, where phase $0$ contains just the first maintenance and phase $i\in\set{1,\ldots,\ceil{\log M}}$ contains maintenances $\set{2^{i-1}+1,\ldots,2^i}$.
The maintenance in phase $i=0$ is performed with failure probability $\delta_0=\delta/2$.
Each maintenance in phase $i>0$ is executed with failure probability $\delta_i = \delta\cdot 4^{-i}$ (e.g., maintenances 3-4 run with $\delta_2=\delta/16$, maintenances 5-8 run with $\delta_3=\delta/64$, etc.).
Using the union bound, one can verify that the overall failure probability is at always lower than $\delta$.

Let us calculate the number of samples this process results in. In phase $0$, the maintenance uses $\mathfrak c_{\sys}\cdot\log(2\delta^{-1})$ samples.
%
%
%
Therefore, if $M=1$, 
we use $\mathfrak c_{\sys}\cdot \log(2\delta^{-1})$, i.e., only $\mathfrak c_{\sys}$ samples over the optimum.
Assume that $M>1$. Each phase $i\in\set{1,\ldots,\floor{\log M}}$ has $2^{i-1}$ maintenances, while the last phase (if $M$ is not a power of two) has $M-2^{\floor{\log M}}$ maintenances.
Therefore, the overall number of samples is:
$
   \mathfrak c_{\sys}\cdot \Bigg( \log(2\delta^{-1}) + 
   \sum_{i=1}^{\floor{\log M}} 2^{i-1}\cdot \log \parentheses{ \frac{1}{\delta\cdot 4^{-i}}}    + (M-2^{\floor{\log M}})\cdot\log\parentheses{\frac{4^{\floor{\log M}+1}}{\delta}} \Bigg)\\
  = \mathfrak c_{\sys}\cdot \Big( M\cdot \log\delta^{-1}   -2^{\floor{\log M}+2}  + M\cdot\parentheses{2\floor{\log M}+2}  + 3\Big)   \le \mathfrak c_{\sys}\cdot \Big( M\cdot \log\delta^{-1} + 2M\log M - 1.8M + 3\Big).
$

That is, our approach always samples at most $2\cdot \mathfrak c_{\sys}\cdot M\log (M/\delta)$ (i.e., less than twice) more than a hypothetical algorithm that knows $M$ at the start of the measurement.

To further improve this algorithm, we can use domain knowledge to ``guess'' the number of maintenance operations $M_0\ge 1$ (where $M_0=1$ coincides with the above approach).
We then have maintenance operations $\set{1,\ldots,M_0}$ in the phase $0$, each runs with $\delta_0=\delta/(2M_0)$. Similarly, each phase $i>0$ contains maintenance ops $\set{M_0\cdot 2^{i-1}+1,\ldots,M_0\cdot 2^{i}}$, each configured for $\delta_i=\delta\cdot 4^{-i} / M_0$.
Using $M_0>1$ reduces the number of phases to one if $M_0>M$ and to $1+\ceil{\log(M/M_0)}$ otherwise.
If $M_0>M$, this gives a bound of $c_{\sys}\cdot M\cdot \log(M_0/\delta_0)=c_{\sys}\cdot M\cdot \log(2M_0/\delta) $ on the total number of samples.
Otherwise, we can bound the overall number of samples by
{\small
$
  \mathfrak c_{\sys}\cdot M_0\cdot \Bigg( \log(2M_0\delta^{-1}) 
  + \sum_{i=1}^{\floor{\log (M/M_0)}} 2^{i-1}\cdot \log \parentheses{ \frac{M_0}{\delta\cdot 4^{-i}}}
  + (M-2^{\floor{\log (M/M_0)}})\cdot\log\parentheses{\frac{M_0\cdot 4^{\floor{\log (M/M_0)}+1}}{\delta}} \Bigg)\\
\ifdefined\verboseAnalysis  
    = \mathfrak c_{\sys}\cdot \Bigg( \log(M_0/\delta) + 1 + (2^{\floor{\log (M/M_0)}}-1)\cdot \log(M_0/\delta)
   + 2^{\floor{\log (M/M_0)}+1}\cdot\floor{\log (M/M_0)} - 2^{\floor{\log (M/M_0)}+1} + 2
   + (M-2^{\floor{\log (M/M_0)}})\cdot\log\parentheses{\frac{4^{\floor{\log (M/M_0)}+1}}{\delta}} \Bigg)\\
  = \mathfrak c_{\sys}\cdot \Bigg( 2^{\floor{\log (M/M_0)}}\cdot \log(M_0/\delta)
  \\ + 2^{\floor{\log (M/M_0)}+1}\cdot(\floor{\log (M/M_0)}-1) + 3
  \\ + (M-2^{\floor{\log (M/M_0)}})\cdot\parentheses{\log\delta^{-1}+2\floor{\log (M/M_0)}+2} \Bigg)\\
  = \mathfrak c_{\sys}\cdot \Bigg( M\cdot \log(M_0/\delta)
  \\ + 2^{\floor{\log (M/M_0)}+1}\cdot(\floor{\log (M/M_0)}-1)  + 3
  \\ + (M-2^{\floor{\log (M/M_0)}})\cdot\parentheses{2\floor{\log (M/M_0)}+2} \Bigg)\\ 
\fi  
  = \mathfrak c_{\sys}\cdot \Big( M\cdot \log\delta^{-1} + M\log M_0
   -2^{\floor{\log (M/M_0)}+2} 
    + M\cdot\parentheses{2\floor{\log (M/M_0)}+2}  + 3\Big).
  %
  %
$
}
{
That is, if $M_0\in[M,M^2]$, our algorithm makes at most $\mathfrak c_{\sys}\cdot \parentheses{ M\cdot \log\delta^{-1} + 2M\log M}$ samples, i.e., always at least as good as $M_0=1$.
Also, if $M>M_0$ we always improve over selecting $M_0>1$. 
We thus conclude that this approach is suitable when our domain expertise allows us to choose $M_0>1$ such that $M\ge \sqrt{M_0}$.
}

\subsection{Score-based Caching on a P4 Switch}\label{subsec:squid-cache}

Score-based caching encapsulates a variety of caching algorithms such as LRU, LFU, SLRU~\cite{SLRU}, LIRS\cite{LIRS-Gil}, and FRD~\cite{FRD}. 
For a given cache size $\mathfrak q$, a score-based caching algorithm is defined using a function that assigns a \emph{score} to each element in the working set based on its request pattern since it was (last) admitted. For example, the LFU cache counts the number of arrivals since an item was cached, while the LRU score is the last accessed timestamp. Intuitively, when a cache miss occurs, the algorithm replaces the lowest-score item with the newly admitted one.

This section shows how \sys can enable score-based caching directly on a hardware switch. Our work is motivated by a new generation of switches that are programmable (e.g., using the P4 programming language), while having similar high throughput and low latency to their non-programmable counterparts~\cite{bosshart2013forwarding, Tofino, trident}.
We show that any such score-based policy is supported, as long as the score can be calculated in the switch's data plane.
We use LRFU~\cite{LRFU} as a working example as we uncover the challenges and \mbox{solutions used to achieve a performant P4 implementation. }

Unlike previous works that use the backend to implement the caching policy~\cite{netcache,switchkv}, \sys decides which items to admit and evict directly in the data plane, thus reducing the latency for cache misses and computational burden on the backend server. Namely, using $q(1+\gamma)$ entries, our goal is to guarantee that the $q$ elements with the highest score are never evicted. The backend is then only used for cache misses to retrieve the value, while the switch's CPU periodically performs \sys's maintenance operations for calculating new approximate quantiles. 




\smallskip
\noindent\textbf{Why P4-based \sys.} 
In-network caching solutions such as NetCache~\cite{netcache}, leverage the flexibility and the high performance (Tera-bit level bandwidth and nanosecond level latency) of programmable switch ASICs to cache frequent items in switches. Such a caching medium can potentially reduce access latency for hot items and the consumed bandwidth. However, these solutions rely on the backend server updating the cache policy to reach eviction and admission decisions, thus increasing the latency on cache misses. Namely, when a cache miss occurs, in addition to providing the value, the backend needs to decide whether to admit the current (key,value) pair and who to evict instead.
Such a dependency creates a possible performance bottleneck for these solutions as the performance of cache policies in software is considerably slower than that of key/value stores or physical switches. Further, NetCache requires the switch to frequently report to the backend the hottest items. This requires additional bandwidth and running a heavy items detector (e.g., a Count-Min sketch~\cite{CountMinSketch}) on the switch, requiring additional resources. 


\textit{Proposed solution.}  
We thus propose that the data plane shall decide which items to evict and admit, according to a score-based caching policy. Our solution achieves the goal by assigning \textit{scores} (determined by the caching policy) for items and applying \sys periodically to compute the water level threshold of the "hottest" $q$-items (\revise{in the control plane as explained later}), regarding items below the water level as evictable (logically deleted, as in \sys-HH) in the following round. 
Notably, the abstraction of score-based caching allows us to describe a wide spectrum of caching policies, that cover a variety of workloads and systems. In contrast, the existing solutions offer ad-hoc cache \mbox{policies that fit the bottlenecks of a specific system~\cite{netcache}. }

Unlike previous solutions, \sys makes admission decisions on cache misses; that is, if the user queries for an uncached key, the answer will be served by the backend, but the switch intercepts the response and can (in the data plane) decide to cache the item.
This data-plane-based cache-update approach transfers the cache-update computation burden from the backend to the P4 switch and reduces the update and cache-miss delay. In Section~\ref{sec:evaluation} we show that our approach supports a wide spectrum of score-based caching policies that are implementable in P4 switches, including LRFU and LRU. We note that \sys's solution involves sending sampled items to backends. However, such an $O(Z)$ overhead (as shown in Table~\ref{tab:p4-traffic-overhead}) per epoch is negligible compared with state-of-the-arts~\cite{netcache, switchkv}. Moreover, we argue that \sys is the enabler of the data-plane-based cache-update solution. Existing solutions such as \qMAX are unlikely to directly be applied for in-network caching, as they typically involve a large communication and CPU-processing overhead. 
For example, using \qMAX may require switches to send $q(1 + \gamma)$ items to their CPUs, which may be prohibitive due to the CPUs' \mbox{throughput being lower than the request rate.}

\noindent\textbf{Challenges of implementing \sys in P4.} The main implementation challenge lies in the limited computational ability of P4 switches. Specifically, P4 only allows memory access to an array within a single stage; within each stage, P4 switches can only perform a series of stateless operations that do not depend on each other or simple stateful operations on an array. Packet recirculation is one workaround solution for these restrictions, but it greatly increases the bandwidth overhead. Without packet recirculation, it is impossible to implement \sys's operations such as pivot computation (finding the $k$'th smallest among the sample), and cuckoo hash table insertions.

\noindent\textbf{The P4 implementation of \sys.} \revise{We propose workarounds to fit it into the P4 switches:}
\begin{enumerate}[align=left, leftmargin=0mm, labelindent=.1\parindent, listparindent=.1\parindent, labelwidth=1mm, itemindent=!,itemsep=2pt,parsep=0pt,topsep=0pt]
    \item We move the pivot value computation task into the control plane. The data plane samples packets and sends them to the controller for calculating the pivot value, and the controller sends the new pivot value back to the data plane. \revise{Since only a small number of items (e.g., $Z=1000$) is sampled, the computation and bandwidth overheads are negligible (see Table~\ref{tab:p4-traffic-overhead})}. Indeed, our \sys's sampling technique is the enabler for the task of ensuring that the top-score $q$ elements are cached -- the original \qMAX solution cannot be implemented like this as finding an exact percentile requires transmitting all the cached items to the controller.
    \item Implementing a hash table is expensive in P4 and we use the considerably simpler array-of-buckets data structure (e.g., see~\cite{randomized-admission, tirmazi2020cheetah,narayana2017language}). Each element is hashed into one of the $c$ subarrays containing $r$ slots. Item admission and eviction are restricted to the slots of a single bucket. 
\end{enumerate}

\noindent\textbf{Implementation details in Tofino switches.} We have implemented P4-based \sys in Intel Tofino and Tofino2 switches in $\sim 1300$ lines of code. The code can be found in our Github repository~\cite{SQUIDopenSource}. We vary the array size as shown in Table~\ref{tab:p4resource}. The P4 programs corresponding to $(r=4, c=2^{15})$ and $(r=4, c=2^{16})$ are implemented in Tofino, the latter one standing for the maximum possible array size implementable in Tofino.
Programs corresponding to $(r=8, c=2^{16})$ and $(r=12, c=2^{16})$ are compiled in Tofino2, which has more stages than Tofino to support larger bucket sizes. We use the stateful register arrays to store the KV pairs and the score of the items. Specifically, for each row of items, we place their keys and scores in one register array in the ingress pipeline and place their values combined with those from its adjacent row in one register array in the egress pipeline. Therefore, theoretically, as a large enough $r$, the number of stages is $\approx\frac{3}{2}r$.

\noindent\textbf{Resource overhead of \sys in Tofino.} Table~\ref{tab:p4resource} shows the resource usage of our Tofino P4 prototype. Regarding stage usage, when there are $r=4$ slots in a bucket, a minimum of $10$ stages should be allocated, which can be compiled in Tofino. The number of stages required is larger than $\frac{3}{2}r$, since the P4 program is bottlenecked by components other than the \sys array, \eg, counters, routing tables, hash number generators, etc. When $r=8$, a minimum of $14$ stages is required, which is supported by Tofino2. \Wenchen{Deleted redundant text.}
Increasing $r$ to $12$ leads to usage of $20$ stages, which also compiles on Tofino2 but leaves fewer resources for other functionalities.
Here the bottleneck becomes the register arrays for storing cached items, since one stage cannot \mbox{accommodate two $2^{16}$ 64-bit register arrays.}

As $r$ increases, we also find that the overheads of other resources such as SRAM, Table IDs, Ternary Bus, etc., grow larger. Since our application is memory-intensive, \sys consumes a large proportion of SRAM memory. 
\ifarXiv
In addition, when $r=12$, the table ID and Ternary Bus usages also reach above $50\%$ and $40\%$, respectively. 
\fi
Note that the TCAM usage is below $3\%$, indicating that our P4 application can be installed \textit{in the same pipeline} in parallel with other TCAM-heavy but SRAM-light applications. Moreover, our \sys P4 application only occupies one pipeline, meaning that other P4 applications can be installed in the other pipelines to run in parallel with \sys.

\begin{table}[H]
\vspace{-0.1in}
    \begin{minipage}[t]{0.751\linewidth}
        \resizebox{1\linewidth}{!}{%
        
        \begin{tabular}{|c||c|c|c|c|c|c|c|} \hline
            $\bm {r \times c}$ & \textbf{Target} & \textbf{Stages} & \textbf{SRAM} & \textbf{TCAM} & \textbf{Table IDs} & \textbf{Ternary Bus} & \textbf{Hash Dist} \\ \hline\hline
            $4 \times 2^{15}$ & \multirow{2}{*}{Tofino} &  10 & 15.29$\%$ & 2.778$\%$ & 25.52$\%$ & 19.79$\%$ & $19.4\%$ \\ \hhline{-~------}
            $4 \times 2^{16}$ &  & 10 & 27.92$\%$ & 2.778$\%$ & 25.52$\%$ & 19.79$\%$ & $20.8\%$ \\ \hline \hline
            $8 \times 2^{16}$ & \multirow{2}{*}{Tofino2} &  14 &  33.20$\%$ & 1.875$\%$ & 36.88$\%$ & 32.50$\%$ & $22.5\%$ \\ \hhline{-~------}
            $12 \times 2^{16}$ &  & 20 & 49.22$\%$ & 1.875$\%$ & 50.63$\%$ & 45.94$\%$ & $30.8\%$ \\ \hline
        \end{tabular}
        
        }

        \caption{\revise{Resource consumption of Tofino \sys prototype for caching.}}\label{tab:p4resource}
        
    \end{minipage}
    \hspace{-2mm}
    \begin{minipage}[t]{0.245\linewidth}
        \resizebox{1\linewidth}{!}{
        \begin{tabular}{|c || c | c | c | c | c | c } \hline
        $\gamma$ & Min speedup  &  Max speedup   \\
        \hline \hline
        $1\%$  & x4.54 & x6.666 \\
        \hline 
        $5\%$  & x2.127 &  x2.5\\
        \hline
        $10\%$  & x1.612 &  x1.923 \\
        \hline 
        $25\%$  & x1.55  &  x1.98 \\
        \hline
        $50\%$  & x1.398  &  x1.709 \\
        \hline
        $100\%$  & x1.142  &  x1.538  \\
        \hline
       \end{tabular}
       }
       
       \caption{\resizebox{.658702638\linewidth}{!}{Speedup vs \qMAX.}}\label{tab:temps} 
    \end{minipage}
    \vspace{-0.144in}
\end{table}

\vspace{-3.5mm}
\section{Evaluation}
\label{sec:evaluation}
\vspace{-0.5mm}
We defer the evaluation setup to Appendix~\ref{app:setup}.

\iffalse  
         \begin{figure}[t]
        \subfloat[Caida16]
        {\includegraphics[width =0.49\columnwidth]{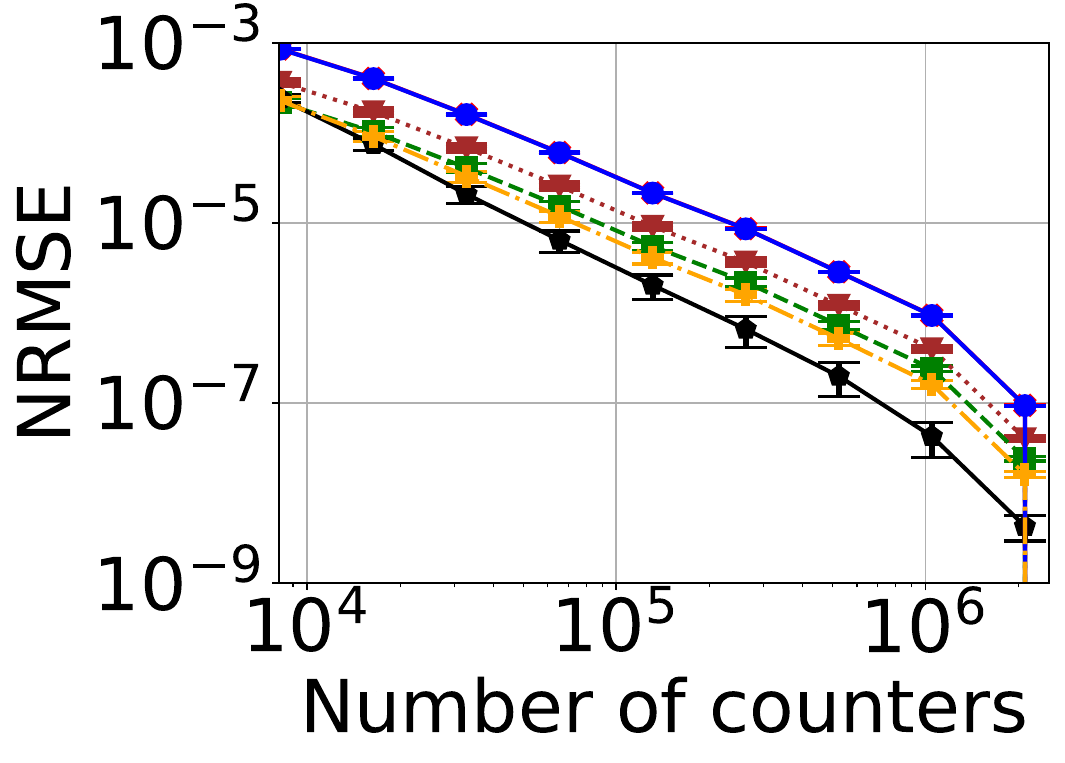}}
         \subfloat[Caida18]
        {\includegraphics[width =0.49\columnwidth]
        {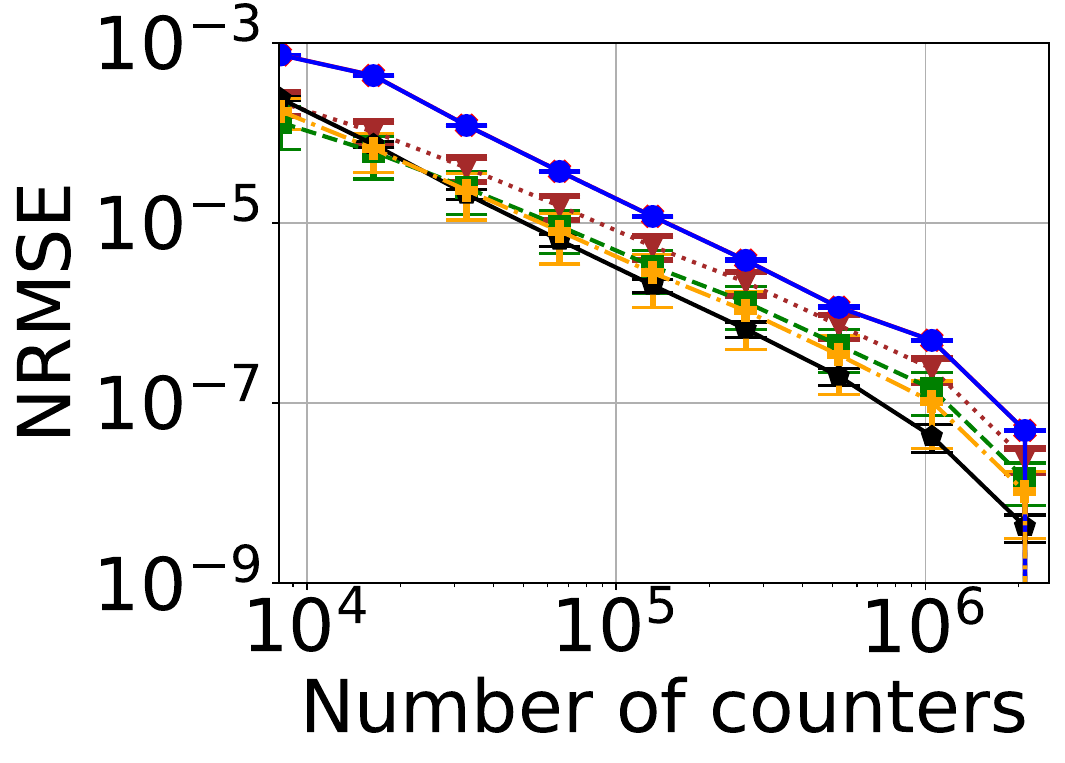}}
        \captionsetup{justification=centering,margin=.2cm}
        {\includegraphics[width =.95\columnwidth]
        {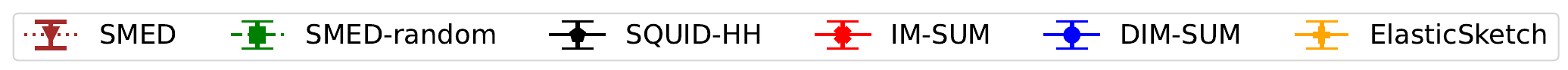}}
        \caption{ The Normalized Root Mean Square Error of SMED, $\mathit{SMED-random}$ and $\sys-HH$ as a function of the number of counters on Caida trace. \label{fig:cucko_err}}
      \end{figure}

        \begin{figure}[t]
        \subfloat[Caida16]
        {\includegraphics[width =0.49\columnwidth]{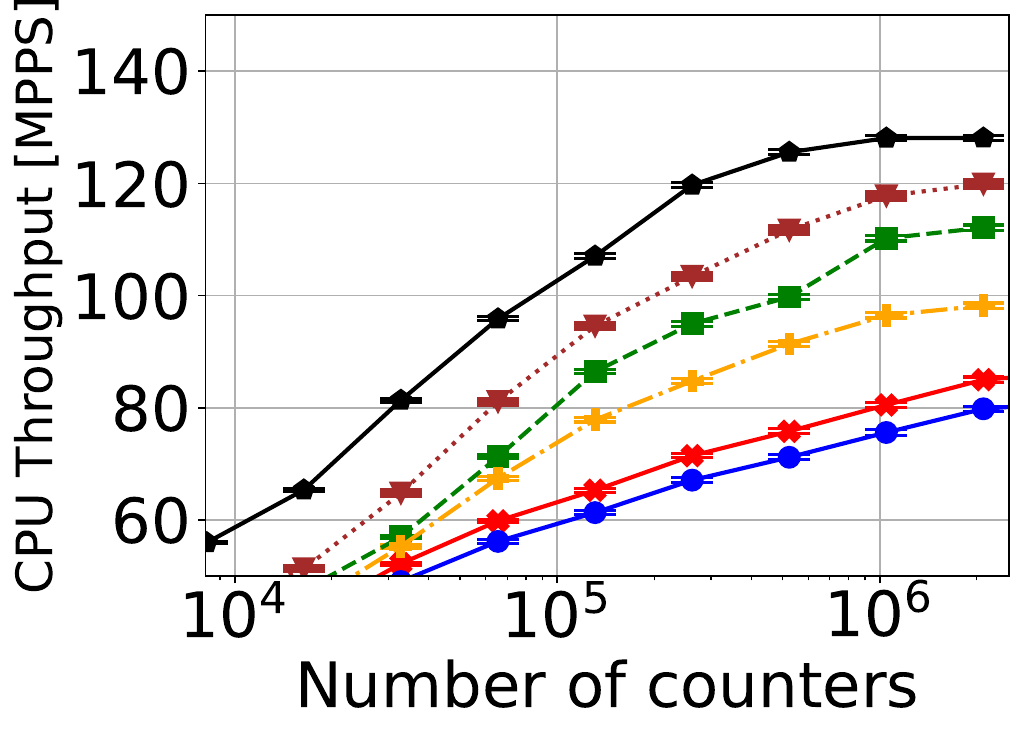}}
        \subfloat[Caida18]
        {\includegraphics[width =0.49\columnwidth]
        {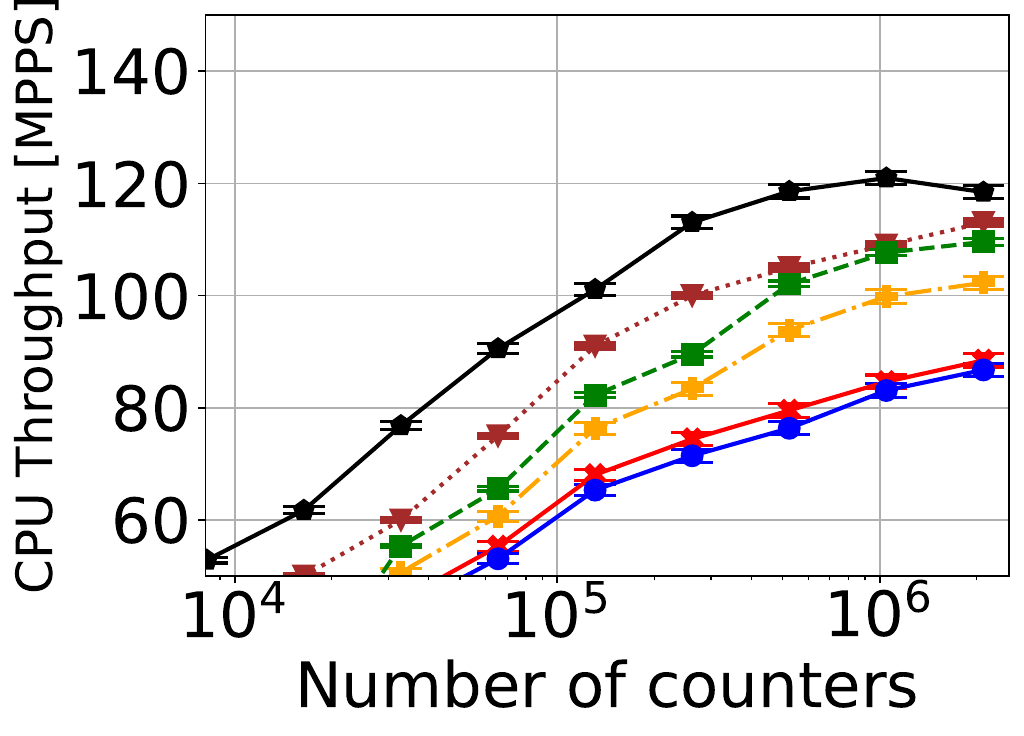}}
        \captionsetup{justification=centering,margin=.2cm}
        {\includegraphics[width =.95\columnwidth]
        {newGraphs/legend1.pdf}}\\
        \caption{ Throughput of SMED, $\mathit{SMED-random}$ and $\sys-HH$ as a function of the number of counters on Caida trace.\label{fig:cucko_Throughput}}
      \end{figure} 
\else

\fi

\vspace{-2.5mm}
\subsection{\sys\ Evaluation}
\vspace{-1mm}
We start by comparing \sys and \qMAX~\cite{qMax}  that work in amortized constant time and to library data structures such as heap and skip-list with logarithmic complexity. 

\begin{figure*}[]
        \vspace{-4mm}
        \subfloat[NWHH $q{=}10^{7}$]
        {\includegraphics[width =0.3\columnwidth]
        {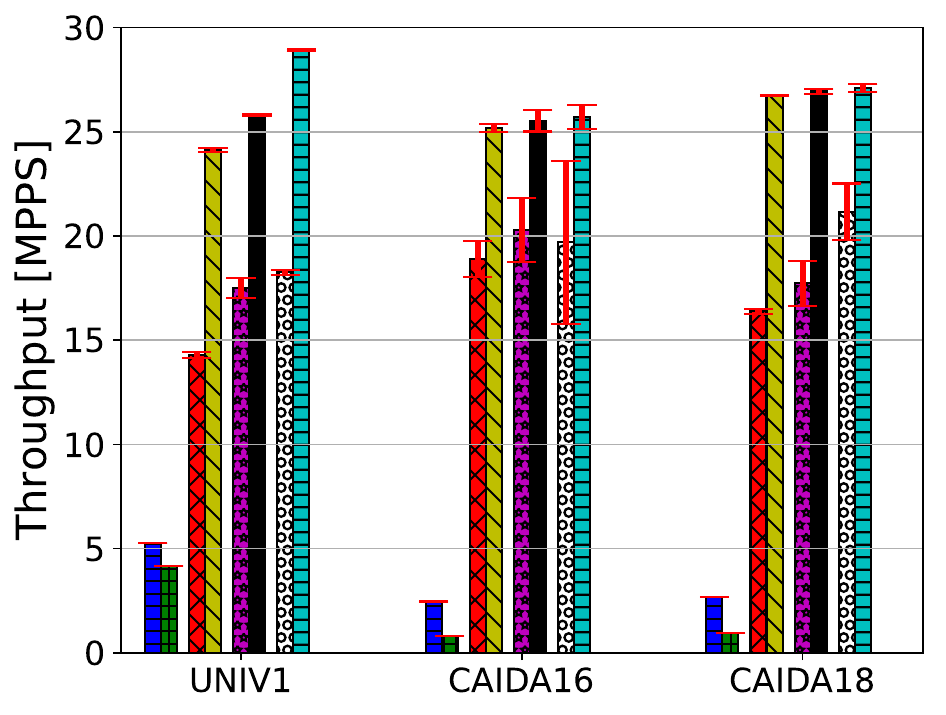}}
        \subfloat[PS $q=10^{7}$]
        {\includegraphics[width =0.3\columnwidth]
        {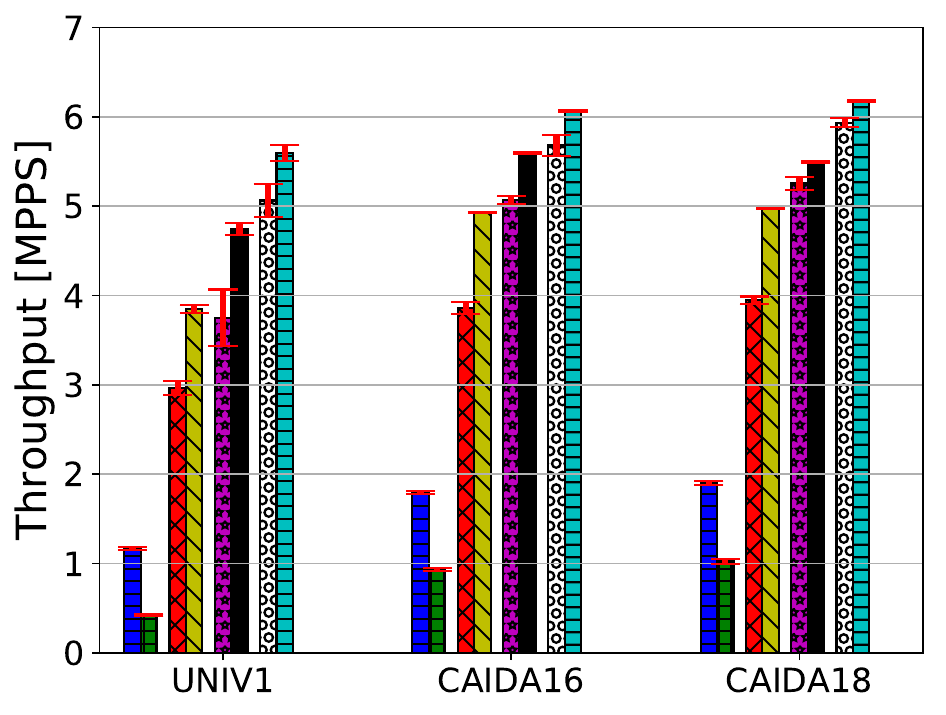}}
        \subfloat[PBA $q=10^{7}$]
        {\includegraphics[width =0.3\columnwidth]
        {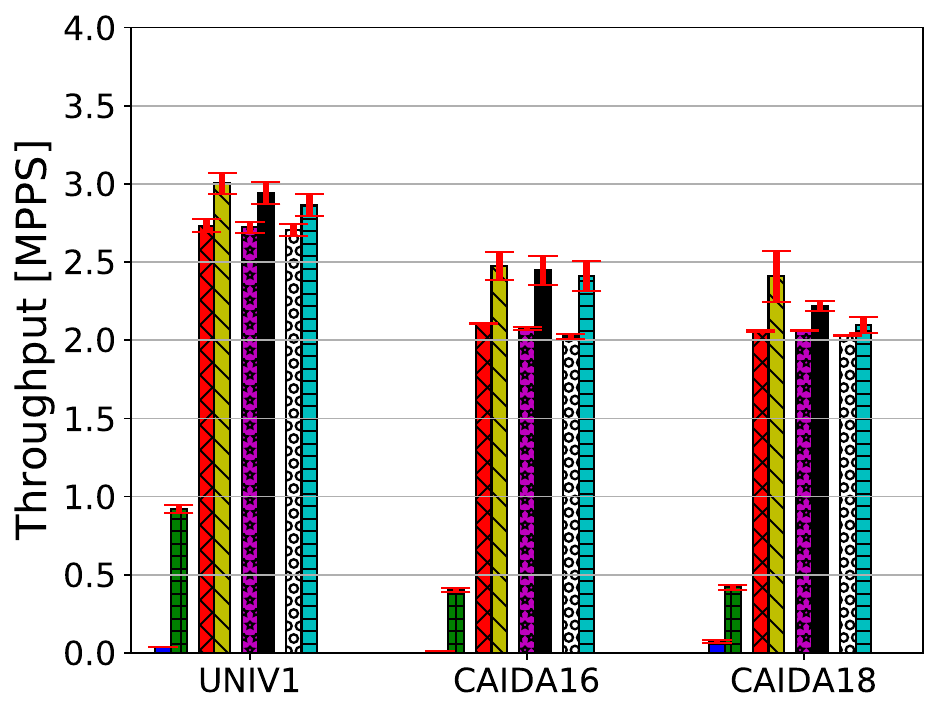}}
        \\
        \vspace{-1mm}
        {\includegraphics[width =.6\columnwidth]
        {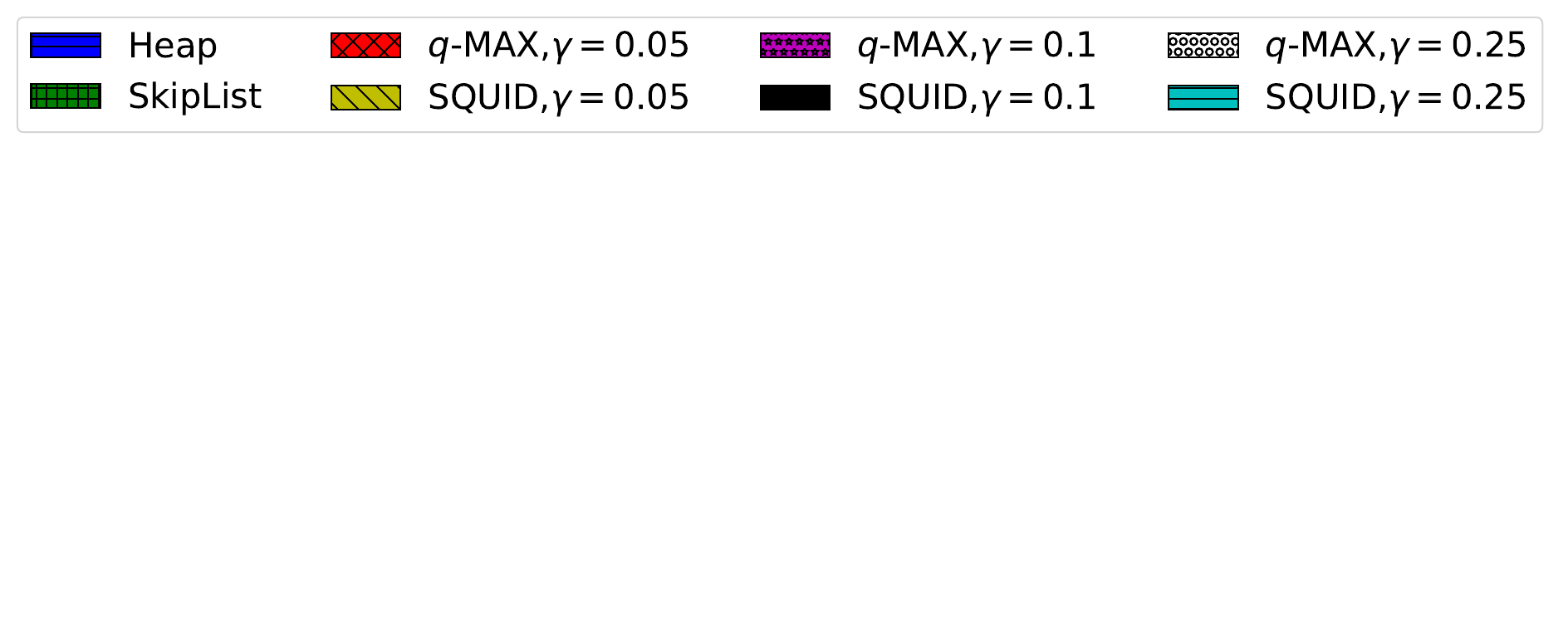}}
        \vspace{-2.8cm}
        \caption{Throughput of various applications when implemented using \qMAX and \sys.}   \label{fig:qmax_application}
        \vspace{-2.05mm}
      \end{figure*}  
\looseness=-1
We evaluate all algorithms using a sequence of 150M random integers, where each algorithm needs to find the $q$ largest of these algorithms for values in $q=10^4,10^5,10^6,10^7$. 
The results, depicted in Figure~\ref{fig:eval}, show that \sys is up to 62\% faster than the best alternative. 
In general, a large $\gamma$ value yields more speedup at the expense of more space, but these are diminishing returns. $\gamma =1$ is double the space, but $\gamma =0.25$ is only a 25\% increase in space.  Observe that  \sys is consistently faster for large $q$ values and similar in performance to \qMAX{} for small $q$ values. To explain this, note that the complexity of finding quantiles does not depend on $q$, while in~\qMAX the complexity is $O(q)$. 
\ifarXiv
In appendix~\ref{app:performance_vs_q}, 
\else
In the full version of the paper, 
\fi
we also evaluate how $z$ affects the performance and show that the speedup is greater for large $q$ values. Table~\ref{tab:temps} summarizes the minimal and maximal speedup when varying $\gamma$ for different q values. As shown, \sys is 1.142-6.666 times faster than \qMAX. Our experiments also indicate that while SQUID without AVX is $6\%-15\%$ slower than with it, it is still markedly (up to 6.4x) faster than q-MAX, i.e., the majority of the improvement comes from the algorithmic advances.

    \begin{figure*}[t]
    \vspace{-4mm}

        \subfloat[Caida16]
        {\includegraphics[width =0.16\linewidth]{newGraphs/NRMSE_caida16_counters}}        
         \subfloat[Caida18]
        {\includegraphics[width =0.16\linewidth]
        {newGraphs/NRMSE_caida18_counters}}
        \subfloat[Univ1]
        {\includegraphics[width =0.16\linewidth]{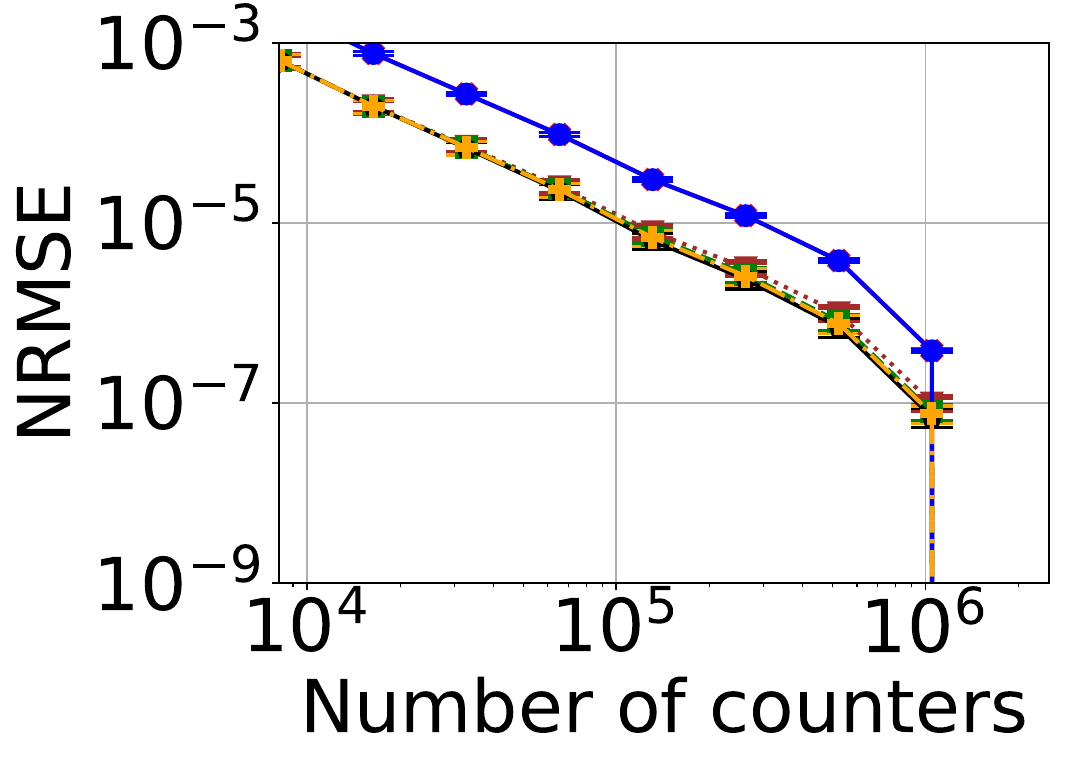}}
        \subfloat[Caida16]
        {\includegraphics[width =0.16\linewidth]{newGraphs/throughput_caida16_counters.pdf}}        
        \subfloat[Caida18]
        {\includegraphics[width =0.16\linewidth]
        {newGraphs/throughput_caida18_counters.pdf}}
        \subfloat[Univ1]
        {\includegraphics[width =0.16\linewidth]{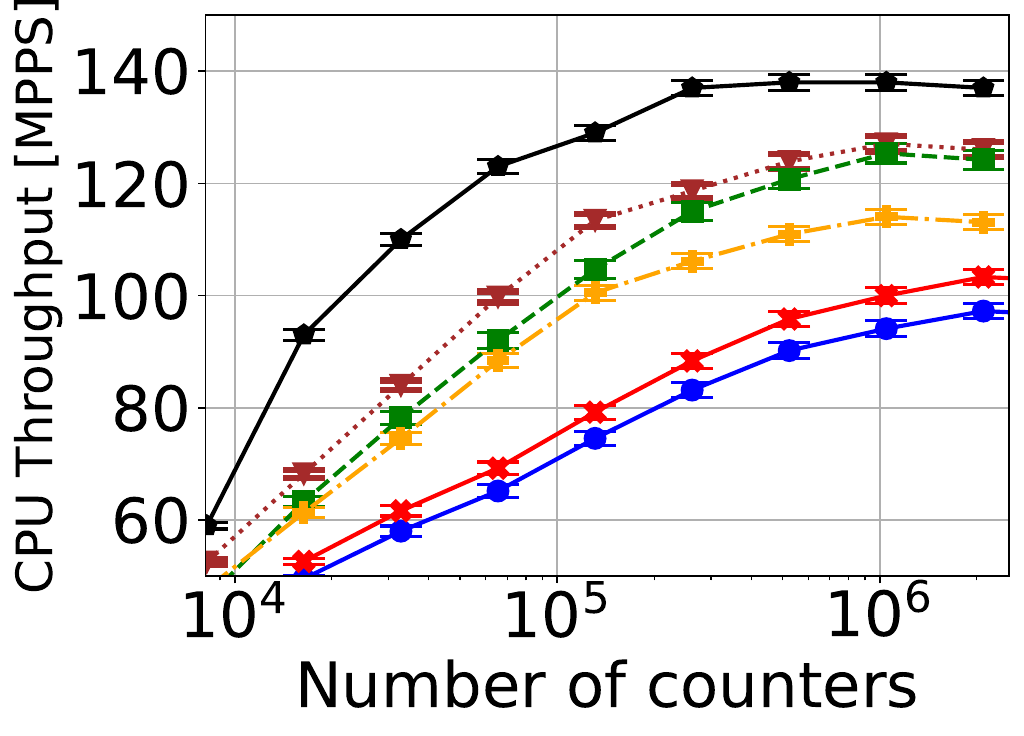}}
        \vspace{-0mm}
        {\includegraphics[width =.8\linewidth]
        {newGraphs/legend1.pdf}}
         \vspace*{-0.135cm}

        \caption{(a)-(c) The Normalized Root Mean Square Error and (d)-(f) the throughput
        of $\mathit{SMED}$, $SMED$-$random$, $\mathit{ElasticSketch}$, and \mbox{$\sys$-$HH$, $DIM$-${SUM}$ and  ${IM}$-${SUM}$, varying the number of counters on packet traces.}\label{fig:cuckoo}}
        \vspace{-3mm}
      \end{figure*}

\iftrue

\else
{\centering
\begin{table}
        \begin{tabular}{c || c | c | c | c | c | c }
        $\gamma$ & Min speedup vs \qMAX  &  Max speedup vs \qMAX  \\
        \hline \hline
        $1\%$  & x4.54 & x6.666 \\
        \hline 
        $5\%$  & x2.127 &  x2.5\\
        \hline
        $10\%$  & x1.612 &  x1.923 \\
        \hline 
        $25\%$  & x1.55  &  x1.98 \\
        \hline
        $50\%$  & x1.398  &  x1.709 \\
        \hline
        $100\%$  & x1.142  &  x1.538  \\
        \hline
       \end{tabular}
    \vspace{1mm}
     \caption{Range of speedup of \sys compared to \qMAX for each value of $\gamma$ on different $q$ values.}
     \label{tab:temps}
\end{table}
}
\fi


\textbf{Measurement Applications:}
In Figure~\ref{fig:qmax_application} We evaluate the impact of replacing \qMAX{}\ with \sys{}\ on the throughput of network applications (Network-wide heavy hitters~\cite{NetworkWideANCS}, Priority sampling~\cite{PrioritySampling}, and priority based aggregation~\cite{PrioritySampling}) using the  Caida16, Caida18 and Univ1 traces, for the following configuration: $q=10^6,10^7$ and $\gamma = 0.05,0.1,0.25$.
The results show that in all the applications, \sys{}  has the highest throughput, and the improvement is more significant for larger $q$ values. For example, in network-wide heavy hitters, we get a speedup of up to $58\%$, and in Priority sampling, the speedup is up to $25\%$, while in priority-based aggregation, we only get up to $14\%$ speedup. The reason for this variability is the relative weight of the top-q calculation in the operation of each algorithm. Specifically, in network-wide heavy hitters, top-q takes a large portion of the computation. Thus, accelerating top-$q$ calculation has a larger impact.

\vspace{-2mm}
\subsection{\sys-HH  Evaluation}
\vspace{-0.5mm}
\label{subsec:expriments}
\looseness=-1
Next, in Figure~\ref{fig:cuckoo} we evaluate \sys-HH ($\gamma=1,w=4$) on both accuracy and speed.
To measure the precision of the algorithms, we use the standard Normalized Root Mean Squared Error (NRMSE) metric~\cite{basat2020faster,basat2021salsa} explained below.
Consider a stream $S=\left\langle{(x_1,w_1),(x_2,w_2),\ldots}\right\rangle$ of packets $(x_n,w_n)$, where $x_n$ is the flow ID and $w_n$ is the packet byte size of the $n$'th packet. Let $f_{x,n}=\sum_{j\le n \mid x_j=x} w$ be the total byte size of flow $x$ up to, packet $n$. Upon the arrival of packet $(x_n,w_n)$, we use the algorithms to estimate the current byte size of $x_n$ and denote the result using $\widehat{f_{x_n,n}}$.
Then:  
$\mathit{NRMSE}=\frac{1}{|S|}\cdot\sqrt{\frac{1}{|S|}\sum_{n=1}^{|S|}\parentheses{f_{x_n,n}-\widehat{f_{x_n,n}}}^2}.$
This gives a single quantity that measures the accuracy of an algorithm over the input stream. We note that NRMSE is a normalized parameter in the $[0,1]$ range, where achieving an NRMSE of $1$ is trivial, e.g., by estimating all flow sizes as $0$.

\looseness=-1
We evaluate \sys and compare it to state-of-the-art algorithms for weighted heavy hitters, such as SMED~\cite{IMSUM}, DIM-SUM~\cite{DIMSUM}, and IM-SUM~\cite{DIMSUM}. The original SMED code includes a runtime optimization that picks the first elements of the array instead of random ones. While this corresponds to a random sample on the first maintenance, the authors did not prove that it produces a random sample in consequent maintenance. Indeed, our evaluation indicates that this optimization improves the throughput at the cost of lower accuracy. For a fair comparison, we evaluate both SMED (that uses the original author's code) and SMED-random, in which we changed the sampling method to take a random subset of counters rather than the first ones and ElasticSketch \cite{Elastic} (using its original code).

As Figures~\ref{fig:cuckoo} (a)-(b) show, \sys is more accurate than all alternatives, thanks to our lazy deletion approach that retains the useful information of logically deleted items as long as possible. Moreso, compare both SMED variants and observe that SMED is less accurate than SMED-random. Thus, not using a real random sample increases the NRMSE. Next notice that IM-SUM and DIM-SUM are the least accurate due to using two hash tables compared to a single table in SMED and \sys{}, Additionally,
Elasticsketch is more accurate than other algorithms but still less accurate than \sys.


As Figures~\ref{fig:cuckoo} (c)-(d) show, \sys is considerably faster than all the alternatives. This is attributed to our water level technique which eliminates the need for a linear pass in the maintenance operation. We also observe that SMED is faster than SMED-random, indicating that the optimization in SMED's code indeed speeds up the runtime. This is due to a more cache-friendly access pattern and fewer random number generations in maintenance operations. Despite being more accurate than SMED and
SMED-random, \mbox{ElasticSketch is slower than both, but still faster than DIM-SUM and IM-SUM.}

\begin{figure*}[]
    \vspace{-3mm}
    \subfloat[CAIDA 2018]
    {\includegraphics[width =0.32\columnwidth]
    {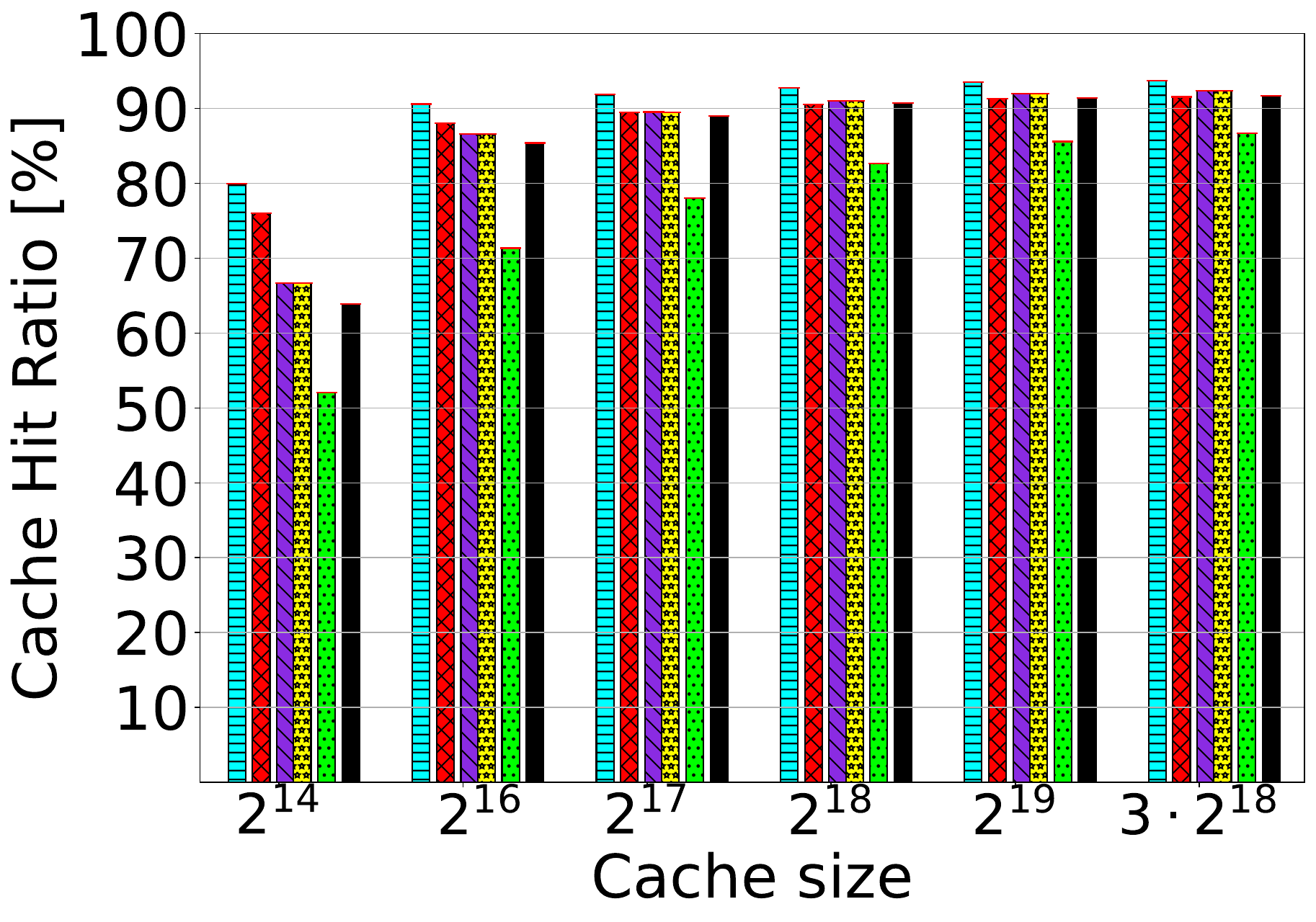}}
    \subfloat[S2 (Caching)]
    {\includegraphics[width =0.32\columnwidth]
    {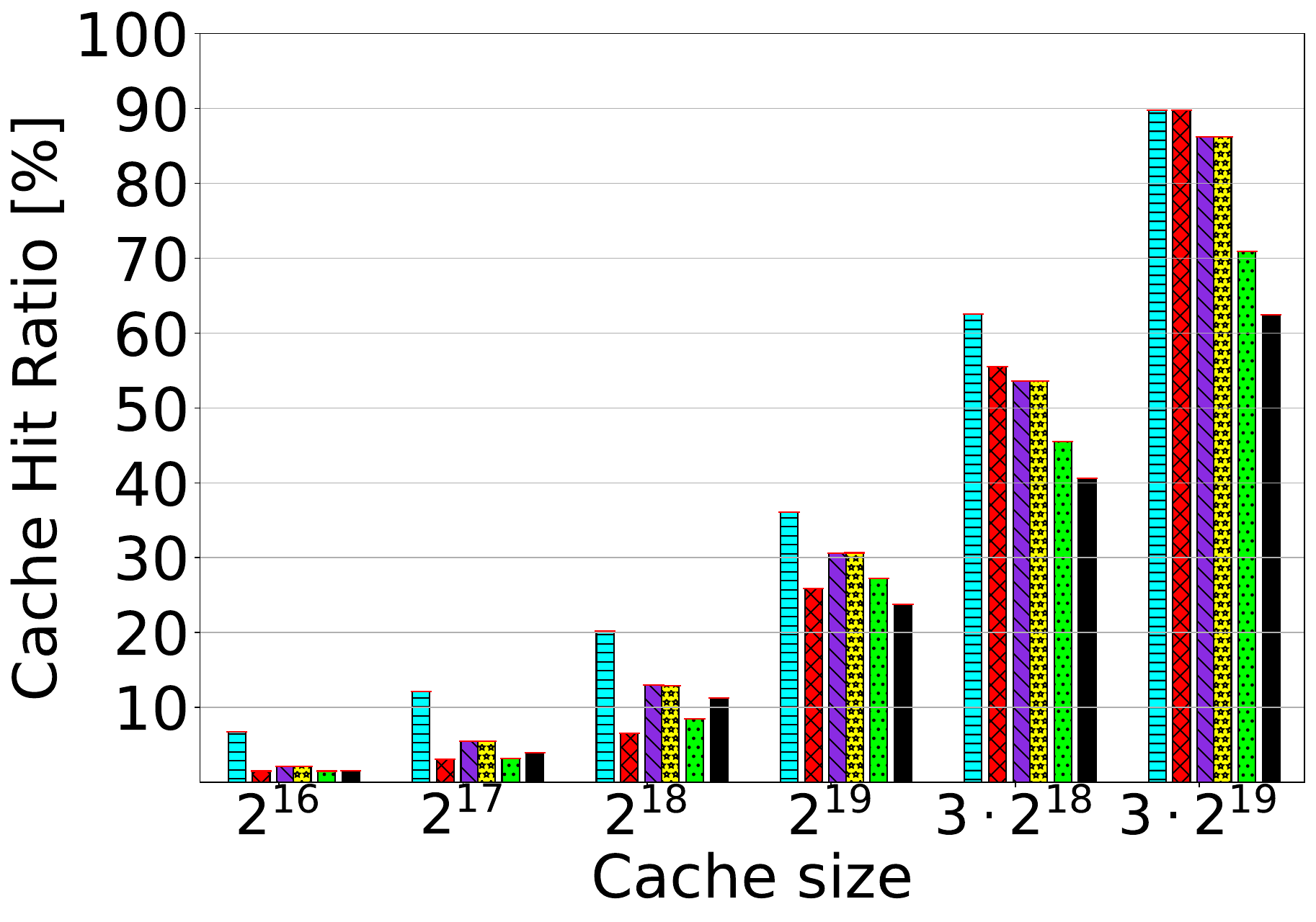}}
    \subfloat[Zipf $0.95$]
    {\includegraphics[width =0.32\columnwidth]
    {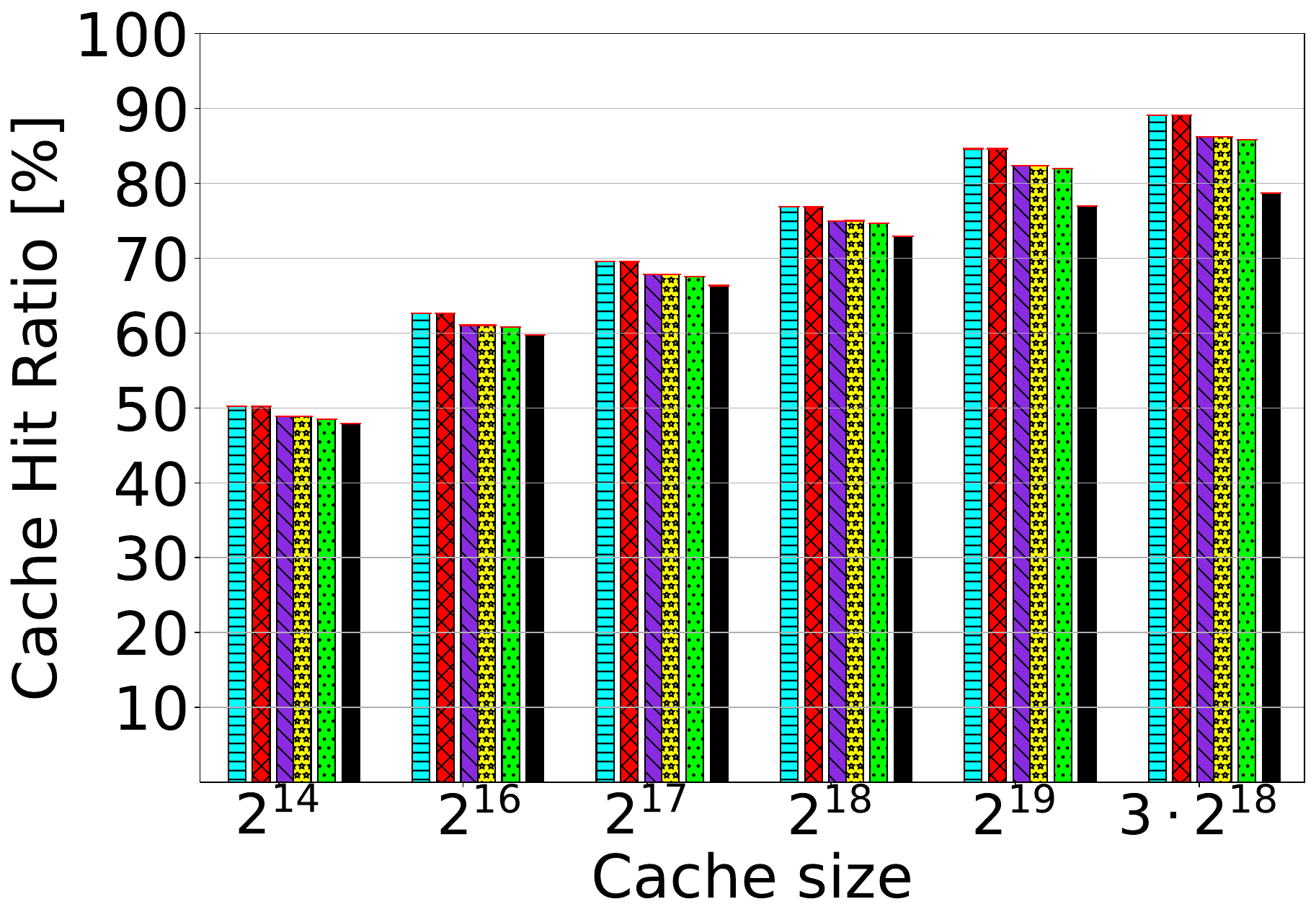}}  \\
    \centering
    {\includegraphics[width =.6132\columnwidth]
    {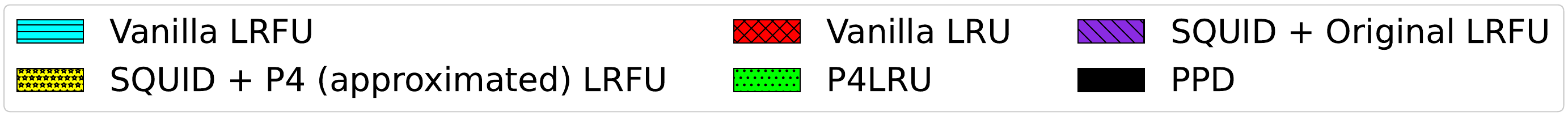}}\vspace{-0.2cm}
    \caption{\revise{Hit ratios of our P4 \sys for in-network caching, compared against vanilla LRFU, LRU, P4-LRU~\cite{P4LRU} and Practical Packet Deflection (PPD)~\cite{ppd}. Both \sys + Original LRFU and \sys + P4 (approximated) LRFU deploy an array-of-buckets structure, \mbox{while Vanilla LRFU maintains a heap to evict global LRFU items.}}}\vspace{-0.1cm}
    \label{fig:p4-lrfu_Hit_Ratio}
\end{figure*}

\begin{table}[]
    \centering
    \resizebox{.7055\columnwidth}{!}{
    \begin{tabular}{|c||c|c|c|c|c|c|c|} \hline
        P4 Cache Implementation  & Stages & SRAM & TCAM & Table IDs & Gateways & Hash Dist \\ \hline
        \sys &  $8$ & $7.85\%$ & $1.88\%$ & $16.2\%$ & $11.5\%$ & $8.3\%$ \\ \hline
        P4LRU~\cite{P4LRU} &  12 & $6.75\%$ & $0$ & $16.2\%$ & $15.0\%$ & $5.8\%$\\ \hline
    \end{tabular}
    }
    \caption{\revise{P4 resource consumption of \sys compared with P4LRU~\cite{P4LRU}. The cache size is set to $r \times c = 3 \times 2^{16}$, which corresponds to the maximum cache size for P4LRU's original implementation.}}
    \label{tab:p4lru-squid-resource}
    \vspace{-5mm}
\end{table}

\subsection{\sys-LRFU Evaluation}

\label{subsec:p4-squid}
\revise{We defer the evaluation of software version of \sys-LRFU, which shows that, as expected, \sys is both faster and with a higher hit-ratio than \qMAX-based LRFU, to Appendix~\ref{app:LRFU}}. 

\revise{Here, we focus on the evaluation of our hardware P4-based \sys implementation for score-based in-network caching. We deploy the LRFU caching policy~\cite{LRFU} and demonstrate how P4 \sys supports a wide range of caching policies. We compare \sys against P4LRU~\cite{P4LRU} and our adaptation of~\cite{ppd} (PPD)'s quantile estimation and show how \sys provides better quantile estimation and outperforms \mbox{their cache hit ratios, thanks to the LRFU generally being more accurate than LRU.}}

\subsubsection{Implementation details of \sys and baselines}\label{subsubsec:p4-impl}\quad\quad\quad\quad\quad\quad\quad\quad\quad\quad

\revise{
\noindent\textbf{Approximate LRFU policy for implementation in Tofino.} One restriction imposed by Tofino P4 architecture is that complex arithmetic operations such as \textit{floating point arithmetic and exponentiation} are not intrinsically supported. We consider the following formulation of LRFU~\cite{LRFU, qMax}. 
Each request $i$ carries a score of $ns_i = -i\cdot\ln c$, where $0.5 \le c\le 1$ is an adjustable parameter. The score update rule is formulated as $s_{i+1} \leftarrow ns_i + \ln(e^{s_i - ns_i} + 1)$, where $s_i$ and $s_{i+1}$ are the score of the cached item before and after the update. \mbox{Such a complex formula cannot be directly computed on P4 switches.}

Therefore, we propose to approximate this LRFU formula. First, we transform it into the following equivalent expression (by scaling up $\times A$ for $A \ge 1$): $s_{i+1} \leftarrow ns_i + A\ln(e^{\frac{1}{A}s_i - \frac{1}{A}ns_i} + 1), ns_i = -A\cdot i\cdot\ln c$. Such transformation allows us to round the scores to integers (that are supported on the switch). Second, to approximate exponentiation and logarithm, we note that 
$ns_i \le max(ns_i, s_i + A\ln 2) \le ns_i + A\ln(e^{\frac{1}{A}s_i -\frac{1}{A}ns_i} + 1)$, where the middle can be computed on the switch. We also note that the LRU policy, which only looks at recency, can be formulated as $s_{i+1}=ns_i$, as it discards any information about past requests. Therefore, by defining the approximate LRFU policy as $s_{i+1} \leftarrow max(-A\cdot i\cdot\ln c, s_i + A\ln 2)$, we get a policy that is somewhat closer to LRU than the standard LRFU interpolation between LRU and LFU. Nonetheless, by varying $c$, we still obtain a spectrum of policies between the two. \mbox{We show that our approximated policy achieves a comparable hit rate to LRFU.}
}


\vspace{-0mm}
\revise{
\noindent\textbf{Baseline systems.} We compare \sys against both P4LRU~\cite{P4LRU}, state-of-the-art LRU cache implementation in the programmable data plane, and PPD~\cite{ppd}'s adaptation to caching. We also implement the original CPU version of Vanilla LRFU and Vanilla LRU strategies for reference.

\noindent (1) P4LRU proposes to maintain an array-of-buckets structure, each containing cached items sorted in LRU order. Due to Tofino pipeline's restrictions~\cite{ben2018efficient}, it is challenging to directly swap the orders of the items' values when the LRU orders change. P4LRU instead proposes to maintain the mapping between keys and values in a bucket as a \textit{permutation}. However, since the number of mappings grows exponentially, \mbox{it cannot support $r > 3$ slots per bucket, limiting the maximum cache size as $3 \times 2^{16}$.}

\noindent (2) To evaluate our sampling-based quantile estimation, we adopt the quantile estimation approach in the Practical Packet Deflection (PPD)~\cite{ppd, aifo}  
to our P4-based \sys's caching design. PPD estimates quantiles entirely in the data plane by sampling every $T$ item and inserting the scores into a small sliding window of size $ns \sim 16$. We leverage \mbox{this to update the water level and evict items accordingly.}
}

\revise{

\subsubsection{Settings}\label{subsubsec:settings} \quad

\noindent\textbf{Datasets}. We conduct experiments on real-world datasets, namely CAIDA~\cite{CAIDA2018} and ARC's caching datasets~\cite{ARC} (S2 and MergeP), \mbox{and on Zipf-synthetic datasets. 
Details are expanded in Appendix~\ref{app:setup}.}

\noindent\textbf{Metrics.} We evaluate both the quantile estimation accuracy (comparing against PPD's quantile estimation) and the cache hit ratio. For the former, we define the Estimation Rank ARE as $(|\hat{r} - r|)/r$, where $\hat{r}$ is the rank of the estimated quantile and $r=q$ is the target quantile rank.

\noindent\textbf{The caching simulator.} We build a behavioral caching simulator to mimic the P4-based caching system, consisting of the P4 data plane (DP), the switch's CPU-based control plane (CP), and a backend server. We set the bandwidth of the DP as $3.2$Tbps per pipeline, CP's processing speed as $100$Kops, the delay between DP and CP as $5\mu$s, and the RTT between the DP and the backend server as $50\mu$s. These \mbox{are consistent with the real-world performance of Tofino-P4 systems~\cite{Tofino, tiara, pcie-performance, hpcc}.} 

\noindent\textbf{Experimental settings of \sys}. As shown in Table~\ref{tab:p4resource}, we let $Z=1000$, and set the length of the register array $c$ to be the power of two, which enables efficient switch implementation. The sizes of the \sys array are set as $4 \times 2^{12}, 4 \times 2^{14}, 4\times 2^{15}, 4\times2^{16}, 8\times 2^{16}$, $12\times 2^{16}$ and $12 \times 2^{17}$ (which simulates running \sys on two Tofino pipelines). 
We assign $\gamma=2$, so that each maintenance can evict enough items 
for later insertion, thereby reducing the frequency of maintenance. 
The $q$ values are then calculated as $q = \frac{r \cdot c}{1 + \gamma} = \frac{r \cdot c}{3}$. To configure LRFU, we 
measure the data distribution of a proportion of the \mbox{earliest items in the backend to tune the parameter $c$, and begin measuring the hit ratio afterward.}

\noindent\textbf{Experimental settings of baselines.} 
    For P4LRU~\cite{P4LRU}, the array sizes ranges from $3 \times \lfloor \frac{2^{14}}{3} \rfloor$, $3 \times \lfloor \frac{2^{16}}{3} \rfloor$ to $3 \times \lfloor 2^{18} \rfloor$.  For PPD~\cite{ppd}, we follow its default settings to assign the sliding window size as $ns=16$, and periodically update the water level (via extra recirculations) every $ns$ samples. We sample once every $50$ items, yielding an extra traffic overhead of $0.125\%$. } 

\begin{table}[]
    \centering
    \resizebox{.6955\columnwidth}{!}{
    \begin{tabular}{|c|c|c|c|c|c|c|} \hline
        $r\times c$ of the register arrays & $4 \times 2^{12}$ &  $4 \times 2^{14}$ & $4 \times 2^{15}$ & $4 \times 2^{16}$ & $8 \times 2^{16}$ & $12 \times 2^{16}$ \\ \hline
        Traffic overheads & $2.8\%$ &  $0.98\%$ & $0.42\%$ & $0.18\%$ & $0.076\%$ & $0.041\%$\\ \hline
    \end{tabular}
    }
    \caption{Control traffic overheads of P4-based \sys for sampling-based quantile estimation. The overheads include DP sending sampled packets to the switch's CP and CP configuring the water level back to the DP.}
    \label{tab:p4-traffic-overhead}
    \vspace{-5mm}
\end{table}

\revise{

\subsubsection{Experimental results}\label{subsubsec:p4-results} \quad


\noindent\textbf{Evaluation on the cache hit ratio.} Figure~\ref{fig:p4-lrfu_Hit_Ratio} and~\ref{fig:p4-lrfu_Hit_Ratio-other} shows the hit ratios of P4 \sys and the baseline approaches, from which we have the following observations. 
First, \sys's caching achieves a better hit ratio than P4LRU, outperforming by up to $16\%$ in both CAIDA and S2 datasets. The main reason is that the LRFU strategies that \sys supports are generally more accurate than LRU, since LRFU protects frequently visited items from being evicted by the random infrequent items. The comparison between Vanilla LRFU and LRU verifies our claim. In the Zipf dataset, the hit ratio differences are negligible, mainly because the dataset is temporally unskewed: each item is drawn independently from the Zipf distribution. Another reason is the deeper slots per bucket ($r$) supported by \sys than P4LRU, so the local LRFU item per bucket better approximates the global LRFU. Second, \sys's hit ratio is also ahead of PPD. 
Such outperformance is due to the much lower quantile estimation error of \sys than PPD, as shown in Figure~\ref{fig:p4-quantile-estimation} and analyzed next. The low quantile estimation error of \sys prevents items above the water level from being marked as evictable and vice versa. Third, Gaps between \sys + Original LRFU and \sys + P4 (approximated) LRFU are negligible. This demonstrates that our arithmetic approximation strategies for LRFU \mbox{successfully simulate the original LRFU policy, making LRFU implementable in Tofino switches.}

\noindent\textbf{Evaluation on the quantile estimation error.} Figure~\ref{fig:p4-quantile-estimation} compares the accuracy of quantile estimation between \sys and PPD's data-plane-based estimatio. We find that \sys achieves approximately one order of magnitude lower quantile estimation error than PPD. 
This is because for \sys, the control plane can accommodate more samples from the data plane for accurate estimation on top of those samples; for PPD, we need to calculate the quantile in the data plane, limiting the number of samples to be visited to $ns=16$ and thus resulting in inaccurate estimation.

\noindent\textbf{P4 resource overheads of \sys and P4LRU}. We also compare the P4 resource consumptions between \sys and P4LRU, as shown in Table~\ref{tab:p4lru-squid-resource}. We set the cache size for both to be $r \times c = 3 \times 2^{16}$, the maximum cache size supported by P4LRU~\footnote{Table~\ref{tab:p4resource} shows that we support a maximum of $12 \times 2^{16}$ entries to fit into Tofino2~\cite{Tofino}.}. 
We observe that the major benefit of \sys is the less usage of the number of stages. 
Additionally, \sys also achieves less gateway usage, which represents resources to support if-else branching in P4. \sys consumes $16.3\%$ more SRAM blocks in order to store the extra LRFU score fields for supporting LRFU. As a reward, \sys achieves a considerably higher hit ratio in datasets like CAIDA and S2.

\noindent\textbf{Control traffic overhead of P4-based \sys.} We evaluate the DP-CP control traffic overhead of our P4-based \sys. 
We define this overhead as the ratio between the traffic for processing cache queries (including both GET request packets and, in case of cache misses, response packets from backend servers) and the extra sampling packets sent to the CP for maintenance. 
As shown in Table~\ref{tab:p4-traffic-overhead}, the overall traffic overhead diminishes as we increase the sizes of the \sys arrays, since larger array size reduces the frequency of maintenance and sampling operations. When the array size reaches $4 \times 2^{16}$, 
the overhead decreases to as small as $0.18\%$. Such a low overhead greatly reduces the traffic burden from the controller and thus makes the controller more reactive to the maintenance operations. Although P4LRU and PPD execute completely in the data-plane, \mbox{we argue that the much higher hit ratio (Figure~\ref{fig:p4-lrfu_Hit_Ratio}) of \sys makes it a favorable solution.}
}

\section{Conclusion}


We introduced the \sys{} algorithm for the micro algorithmic pattern of retaining the $q$ largest items, the \sys-HH for the weighted heavy hitter problem, and the \sys-LRFU for more general score-based caching policies.
Our \sys{} algorithm is faster than standard data structures like Heap and Skip list and from the previously suggested approaches~\cite{qMax}. Such an improvement shows a potential to improve the throughput of numerous network algorithms that utilize the above-mentioned micro algorithmic pattern~\cite{UnivMon,Nitro,BarYosef,Duffield:2017:SAT:3132847.3133042, NetworkWideANCS,Duffield:2007:PSE:1314690.1314696}.  
Specifically, our evaluation demonstrates concrete improvements for Network-wide heavy hitters~\cite{NetworkWide}, Priority Sampling~\cite{PrioritySampling}, and Priority \mbox{Based Aggregation~\cite{PrioritySampling}.} 
Moreso, \sys-HH targets the weighted heavy hitters' problem and is faster and more accurate than the best alternatives~\cite{IMSUM,DIMSUM} when evaluated on real workloads. 

\looseness=-1


Our \sys-LRFU algorithm targets the broader context of score-based caching in software and in P4 programmable switches. In software, We demonstrate a throughput improvement compared to previous implementations~\cite{LRFU,qMax} with a negligible effect on the hit ratio. \revise{In P4 switches, our work implements a broad spectrum of score-based cache policies in the data plane such as LRU, LFU and any policy between. It achieves higher hit ratios than  P4LRU~\cite{P4LRU} and better quantile estimation than PPD~\cite{ppd} with acceptable switch overheads.}
For reproducibility, we \mbox{open source~\cite{SQUIDopenSource} our code.}


\newpage 
\bibliographystyle{abbrv}
\bibliography{paper}

\newpage
\appendix
\onecolumn
\section{Parameter Tuning}\label{app:param}
\paragraph{The Expected Iteration Length}\ \\
The above analysis looks for a bound on the probability of getting an iteration of at least $q\cdot \eta\cdot\gamma$ elements.
In practice, even if we get $X_{(k)} < q\cdot \eta\gamma$ it is still useful to run the iteration rather than draw a new pivot.
To that end, we can approximate $X_{(k)}$ as $q(1+\gamma)\cdot W_k$, where $W_k\sim\mathit{Beta}(k,Z-k+1)$ is distributed like the $k^{th}$ order statistic of $Z$ i.i.d. uniform, continuous, $U[0,1]$ variables.
The conditioned expectation of $X_{(k)}$ can then be expressed as
\begin{align}
\mathbb E[X_{(k)} | X_{(k)} \le q\gamma]\ge q(1+\gamma)\mathbb E[W_k | W_k \le \gamma/(1+\gamma)]-1.\label{eq:expectedLength}
\end{align}

Using the probability distribution function of $\mathit{Beta}(k,Z-k+1)$, we can then write
$$
\mathbb E[W_k | W_k \le \gamma/(1+\gamma)] = \frac{\int_0^{\gamma/(1+\gamma)}x\frac{x^{k-1}(1-x)^{Z-k}}{B(k,Z-k+1)}dx}{\Pr\brackets{W_k\le \gamma/(1+\gamma)}}
= \frac{B(k+1,Z-k+1)\cdot I_{\gamma/(1+\gamma)}(k+1,Z-k+1)}{B(k,Z-k+1)\cdot I_{\gamma/(1+\gamma)}(k,Z-k+1)},
$$

where $B(\cdot)$ is the Beta function and $I$ is the Regularized Incomplete Beta function.

While the expression has no closed-form formula, we can evaluate it numerically. 
For example, the parameters $\gamma=1$, $\alpha=0.8$, $k=304$, and $Z=760$ (calculated for $\delta=0.1\%$) yield $\mathbb E[W_k | W_k \le \gamma/(1+\gamma)]\approx 0.399995$. We note that $\mathbb E[W_k]=k/Z=0.4$; the correct conditioned expectation, $\mathbb E[W_k | W_k \le \gamma/(1+\gamma)]$, is slightly lower since we demand that the pivot is always among the $q\gamma$ smallest elements.
Plugging this back to~\eqref{eq:expectedLength}, we get that the expectation is nearly $0.8\gamma\cdot q$.

\ifarXiv
\begin{table}[]
\resizebox{0.655\columnwidth}{!}{%
\hspace*{-4mm}\begin{tabular}{|l||l|}
\hline
\textbf{Symbol}   & \textbf{Meaning}                                                                          \\ \hline\hline
\textbf{$q$}      & The number of largest elements to track.                                                  \\ \hline
\textbf{$\gamma$} & The amount of space is $q(1+\gamma)$.                                                     \\ \hline
\textbf{$Z$}      & The number of samples from the array.                                                     \\ \hline
\textbf{$k$}      & We pick the $k$'th smallest sample as pivot.                                              \\ \hline
\textbf{$\alpha$} & We choose $k$ so the pivot in expectation is smaller than $q\gamma\alpha$ values.         \\ \hline
\textbf{$\eta$}   & A function of $\alpha$. A pivot is successful if it is larger than $q\gamma\eta$ values. \\ \hline
\textbf{$\delta$} & A bound on the failure probability.                                                       \\ \hline
\textbf{$\epsilon$} & A bound on heavy hitters error.                                                      \\ \hline
\textbf{$c$}      & \sys-HH is using $c/\epsilon$ counters.                                    \\ \hline
\textbf{$w$}    & The width the \sys-HH Cuckoo hash table.       \\ \hline
\textbf{$d$}    & The depth of the \sys-HH Cuckoo hash table ($c=w\cdot d/\epsilon$).       \\ \hline
\end{tabular}
}
\vspace{1mm}
\caption{\ran{Move to an appendix}The notations used in the paper.}\vspace*{-10.5mm}\label{tbl:notations}
\end{table}
\fi

\subsection{Optimizing the $\alpha$ Parameter}
%
To understand our algorithm, we consider the budget of $Z$ samples given and tune our algorithm to optimize the expected number of elements cleared in an iteration.
For that, we get a lower bound on the expected number of elements removed by our iteration of 
\begin{multline*}
\mathbb E[X_{(k)}\mid X_{(k)}\le q\gamma]\cdot \Pr[X_{(k)}\le q\gamma]
\ge 
\parentheses{q(1+\gamma)\mathbb E[W_k | W_k \le \gamma/(1+\gamma)]-1}\cdot\parentheses{1- e^{-k\cdot \frac{(\gamma -\alpha\gamma)^2}{\alpha\gamma\cdot(2+\gamma -\alpha\gamma)}}}\\
=
\parentheses{q(1+\gamma)\frac{B(k+1,Z-k+1)\cdot I_{\gamma/(1+\gamma)}(k+1,Z-k+1)}{B(k,Z-k+1)\cdot I_{\gamma/(1+\gamma)}(k,Z-k+1)}-1}\cdot\parentheses{1- e^{-k\cdot \frac{(\gamma -\alpha\gamma)^2}{\alpha\gamma\cdot(2+\gamma -\alpha\gamma)}}}.
\end{multline*}
Our goal is to maximize this expectation (by setting the right $\alpha$) while fixing the value of $Z$, therefore we express the above as a function of $\alpha$ by setting $k=\alpha\gamma\frac{Z}{(1+\gamma)}$:
\begin{multline*}
F_Z(\alpha)\triangleq
 \parentheses{1- e^{-\alpha\gamma\frac{Z}{(1+\gamma)}\cdot \frac{(\gamma -\alpha\gamma)^2}{\alpha\gamma\cdot(2+\gamma -\alpha\gamma)}}} \cdot\\
\parentheses{q(1+\gamma)\frac{B(\alpha\gamma\frac{Z}{(1+\gamma)}+1,Z-\alpha\gamma\frac{Z}{(1+\gamma)}+1)\cdot I_{\gamma/(1+\gamma)}(\alpha\gamma\frac{Z}{(1+\gamma)}+1,Z-\alpha\gamma\frac{Z}{(1+\gamma)}+1)}{B(\alpha\gamma\frac{Z}{(1+\gamma)},Z-\alpha\gamma\frac{Z}{(1+\gamma)}+1)\cdot I_{\gamma/(1+\gamma)}(\alpha\gamma\frac{Z}{(1+\gamma)},Z-\alpha\gamma\frac{Z}{(1+\gamma)}+1)}-1}
.
\end{multline*}

\looseness=-1
While the above function is unlikely to admit analytical optimization, we can search for the right 
$\alpha$ value using numeric means.

When $Z\gg \frac{1+\gamma}{\gamma^2(1-\alpha)^{2}}$ (and thus $k\gg \frac{\alpha}{\gamma(1-\alpha)^{2}}$), we can approximate $\mathbb E[W_k | W_k \le \gamma/(1+\gamma)]$ as $\mathbb E[W_k]=k/Z$. 
The reason is that the probability that $X_{(k)}$ is larger than $q\gamma$ is small (as evident by~\eqref{eq:failureProb}) and only slightly affects the conditioned expectation $\mathbb E[W_k | W_k \le \gamma/(1+\gamma)]$.
In this case, we can get a simpler optimization function:
{\scriptsize
\begin{multline*}
F_Z(\alpha)\approx
 \parentheses{1- e^{-\alpha\gamma\frac{Z}{(1+\gamma)}\cdot \frac{(\gamma -\alpha\gamma)^2}{\alpha\gamma\cdot(2+\gamma -\alpha\gamma)}}}
\parentheses{q(1+\gamma)\cdot k/Z}
=
 \parentheses{1- e^{-\frac{Z}{(1+\gamma)}\cdot \frac{(\gamma -\alpha\gamma)^2}{(2+\gamma -\alpha\gamma)}}}
\parentheses{q\alpha\gamma}
\ge 
\parentheses{1- e^{-\frac{Z\gamma^2}{(1+\gamma)(2+\gamma)}\cdot {(1 -\alpha)^2}}}
\parentheses{q\alpha\gamma}
.
\end{multline*}
}

The resulting function is concave in the range $\alpha\in(0,1)$ and is easy to optimize numerically.
For example, when $\gamma=1$, using $Z=760$ samples it is maximized at $\alpha\approx 0.83$, 
giving a lower bound of $0.808q\gamma$ elements that are removed on average, \mbox{when factoring in the probability of a bad pivot choice. }
\section{De-amortizing \sys}\label{app:deamortization}
We explained \sys{} as having separate update and maintenance procedures. 
Previous works, such as~\cite{DIMSUM,DIMSUM++,qMax} suggest deamortization approaches that allow the algorithm to perform $O(1)$ time maintenance operations per update while guaranteeing that there will always be some free space.
However, these deamortization procedures assume that the algorithm is deterministic. 
That is, if we are guaranteed that the maintenance terminates within $c\cdot q$ operations for some constant $c$, and each iteration has $q\gamma$ insertions (as in~\cite{qMax}), it is enough to perform $\frac{c\cdot q}{q\gamma} = c/\gamma=O(1/\gamma)$ maintenance operations per packet.
Deamortizing \sys is slightly more challenging since the above Las Vegas algorithm has no bound on the number of maintenance operations, and our iteration length is a random variable itself.
We resolve this issue by running an exact quantile computation if the sampling algorithm fails more than some constant number of times (e.g., two). This way, we can still compute an absolute constant $c'$ such that we make at most $c'\cdot q$ operations and free at least $q\gamma(2\alpha-1)$ elements. Therefore, we can deamortize \sys by making $\frac{c'\cdot q}{q\gamma(2\alpha-1)}=O(1/\gamma)$ \mbox{operations per update, which is constant for any fixed $\gamma$.}

   \begin{figure}[t]
        \subfloat[$q=10^{4}$]{
        \includegraphics[width =0.24\linewidth]{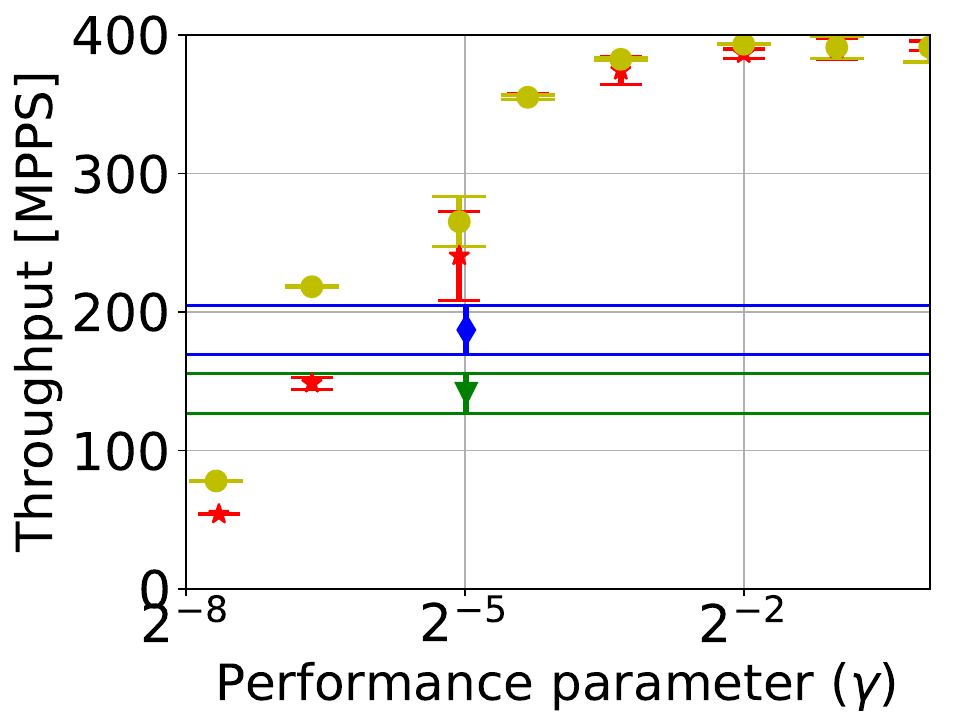}}
        \subfloat[$q=10^{5}$]
        {\includegraphics[width =0.24\linewidth]
        {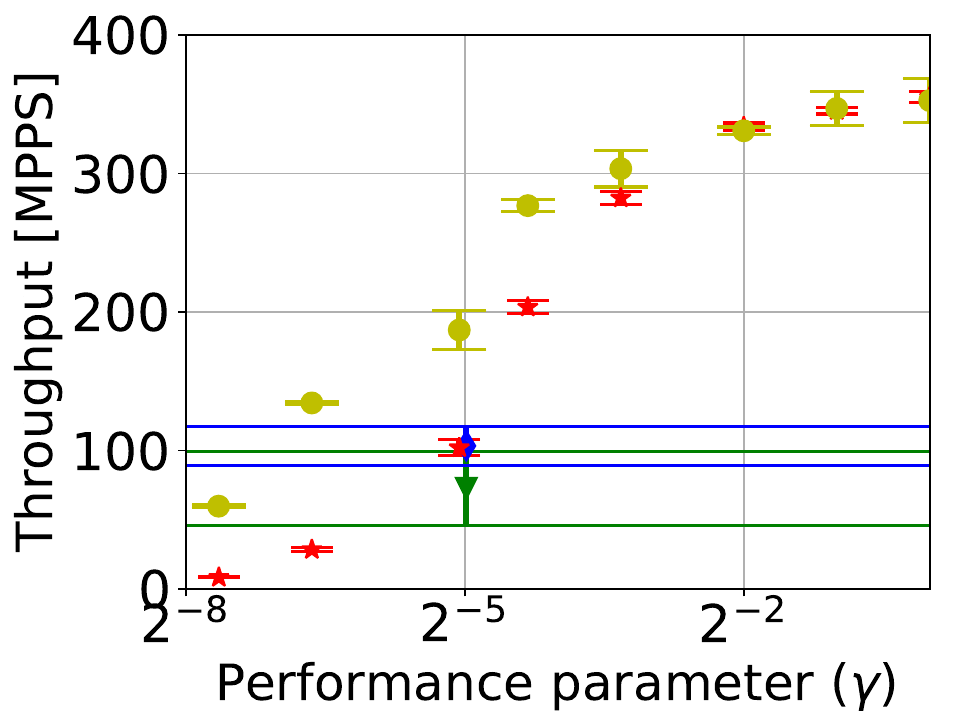}}
        \centering
        \subfloat[$q=10^{6}$]
        {\includegraphics[width =0.24\linewidth]{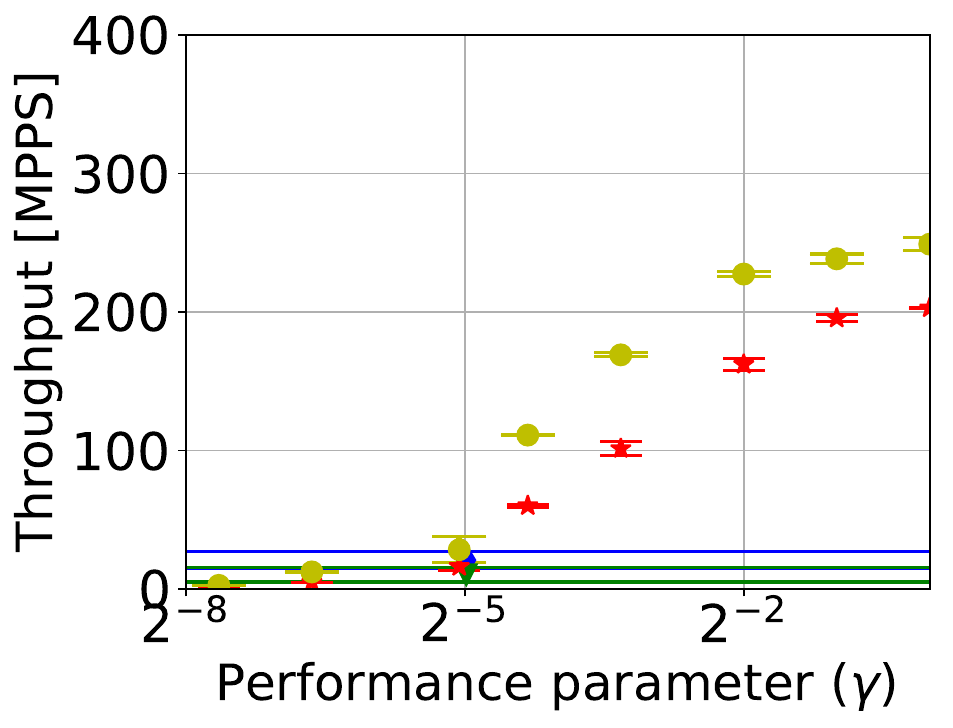}}
        \subfloat[$q=10^{7}$]
        {\includegraphics[width =0.24\linewidth]
        {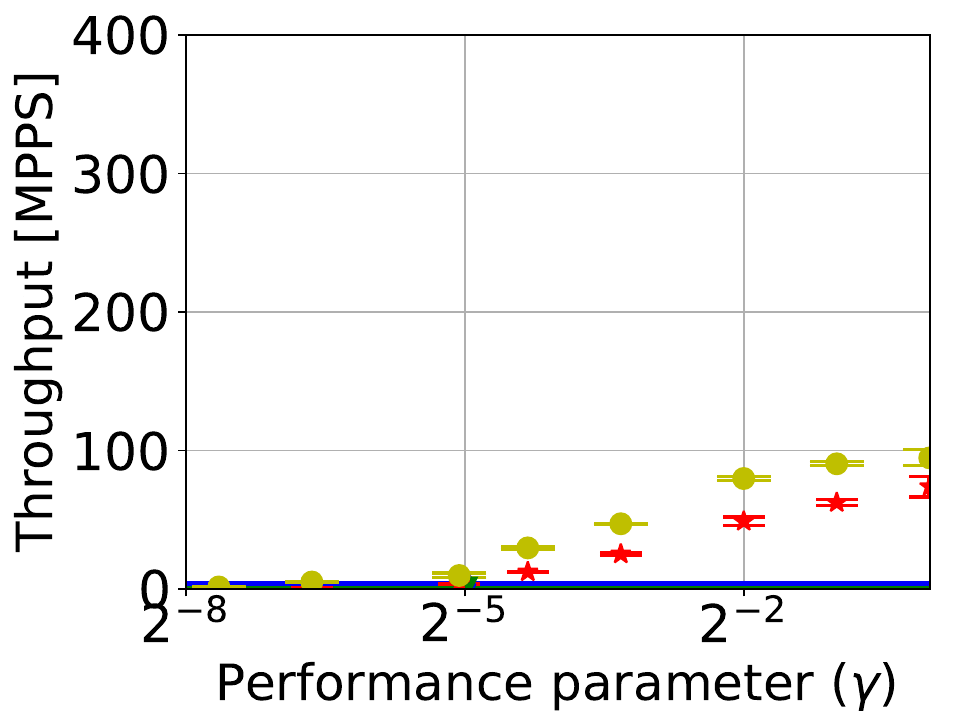}}\\
        \centering
        \vspace*{0.5mm}
        {\includegraphics[width =0.4\columnwidth]
        {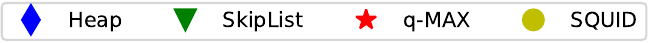}}\vspace{-1mm}
        \caption{Throughput of \sys and \qMAX as a function of $\gamma$ on 150M random numbers. \ran{Move to an appendix and reference it when mentioning table 3.}}\label{fig:eval}\vspace{-3mm}
      \end{figure}

\section{Experimental Setup}

This section positions \sys{}, against the leading alternatives such as \qMAX~\cite{qMax}, as well as to standard data structures such as heaps and skip lists. Next, we compare \sys{}-HH against SMED~\cite{IMSUM}. Finally, we demonstrate that our approach has applications on a broader scope than network monitoring by implementing the LRFU cache policy using \sys{} and comparing our implementation to the existing alternatives. We used the original C++ code released by the authors of competing algorithms~\cite{IMSUM,qMax}, with the recommended parameters, and we implemented our algorithms with C++ as well for a fair comparison. We also implemented the P4-based \sys and \mbox{simulated it in Python for evaluation.}

\noindent\textbf{Datasets:}\label{app:setup}
We used the following datasets:
\begin{enumerate}
\item The CAIDA Anonymized Internet Trace~\cite{CAIDA} (Caida16), from the Equinix-Chicago monitor with 152M packets.
\item The CAIDA Anonymized Internet Trace ~\cite{CAIDA} from New York City (Caida18) with 175M packets.
\item Data center network trace~\cite{UNIV} (Univ1) with 17M packets.
\item Windows server memory accesses trace denoted P1~\cite{ARC}. It has 1.8M accesses and 540K addresses.
\item Memory accesses of a finance application denoted F1~\cite{Umass}. It has 5.4M accesses with 1.4M distinct files.
\item Queries from a popular web search engine denoted WS1~\cite{Umass}.
It has 4.58M queries with 1.69M distinct queries.
\item Accesses taken from the Gradle build cache~\cite{CaffeineProject}.  It has 2.1M accesses with 648K distinct files.
\item  E-commerce memory accesses by Scarab Research~\cite{cache1} \mbox{with 1.94M accesses to 912K addresses.}
\item Memory accesses of an OLTP server of a  financial institution taken from~\cite{ARC} denoted OLTP. It has 4.2M accesses with 1.34M distinct addresses.
\item \revise{Memory accesses of caching applications denoted as MergeP and S2~\cite{ARC}, with $\sim 20$M accesses.}
\item The generated Zipf datasets choosing the Zipf parameter as $0.9$, $0.95$ and $0.99$ (as common in switching caching works~\cite{switchkv, netcache, distcache}). Each dataset has $50$M items with $1.6$M distinct items. These datasets are mainly used for the evaluation of our P4 \sys.
\end{enumerate}

For evaluating \sys{}, and \sys-HH, we used the decimal representation of the IP source address of TCP and UDP packets as the key and the total length field in the IP header as the value. Thus, the evaluation considers the first 5 minutes of CAIDA’16 and CAIDA’18 traces and all the UNIV1 trace.
We ran each data point ten times and report the mean and 99\% confidence intervals.
We run our evaluation on an Intel 9750H CPU running 64-bit Ubuntu 18.04.4, 16GB RAM, 32KB L1 cache, 256KB L2 cache, and 12MB L3 cache. \mbox{Our code is written in C++ and P4 and is available at~\cite{SQUIDopenSource}.}

\ifarXiv
\section{Throughput for different $q$ values}\label{app:performance_vs_q}
              \begin{figure}[h]
              \vspace{-3mm}
        {\includegraphics[width =0.33\columnwidth]{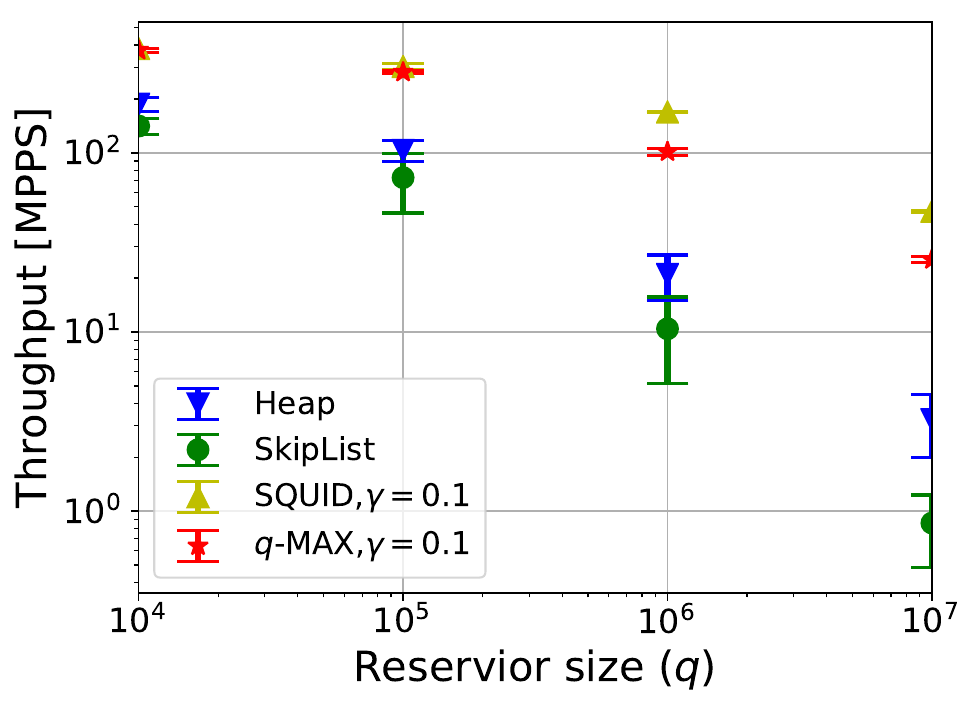}}
        \vspace{-4mm}
        \caption{Throughput of \sys and \qMAX as function of $q$ on a stream of 150M random numbers.\label{fig:reservior}}
        \vspace*{-3mm}
      \end{figure}  
Figure~\ref{fig:reservior} compares \sys to \qMAX for varying $q$ values. As can be observed, \sys{} is at least as fast as \qMAX for all values of $q$ with $\gamma = 0.1$.  \mbox{Moreover, when $q$ is large, \sys is up to $41\%$ faster.   }
\fi



%


        \begin{figure}[]
        \subfloat[P1]
        {\includegraphics[width =0.3\columnwidth]
        {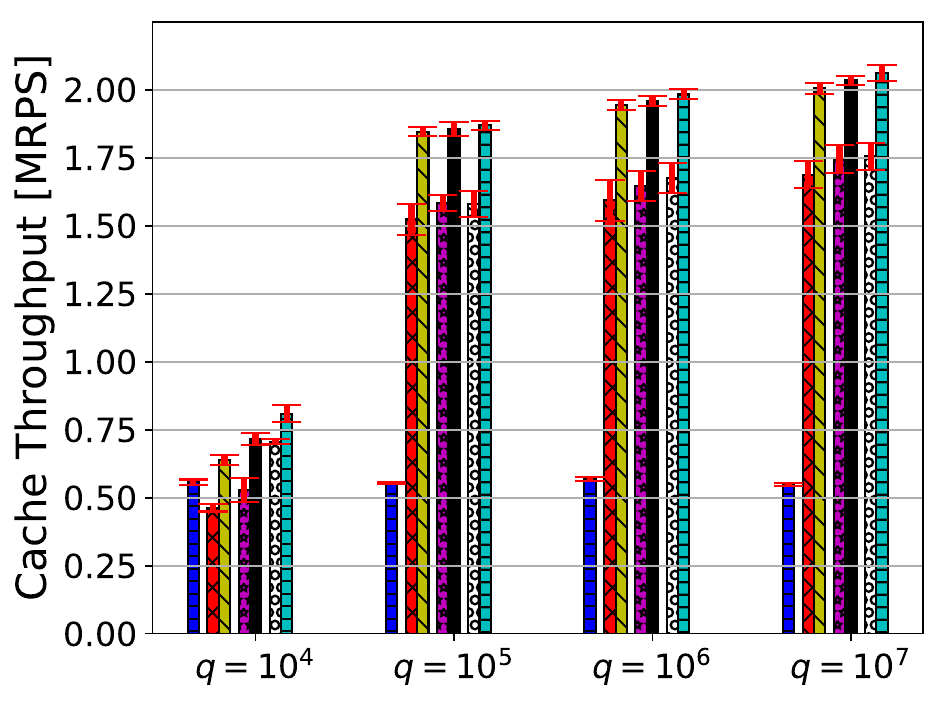}}
        \subfloat[F1]
        {\includegraphics[width =0.3\columnwidth]
        {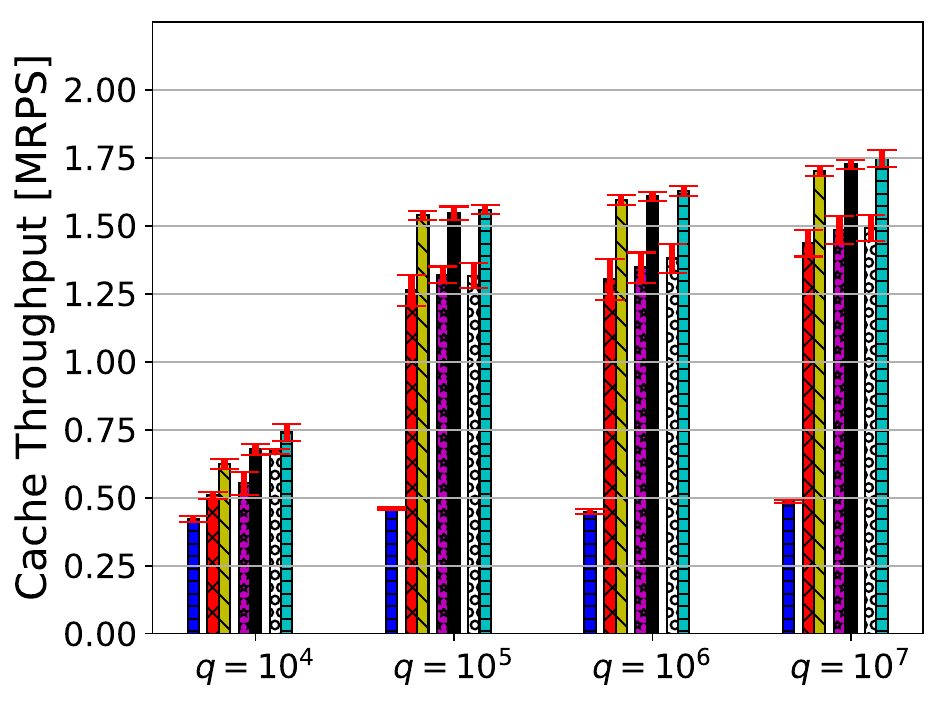}}
        \subfloat[WS1]
        {\includegraphics[width =0.3\columnwidth]
        {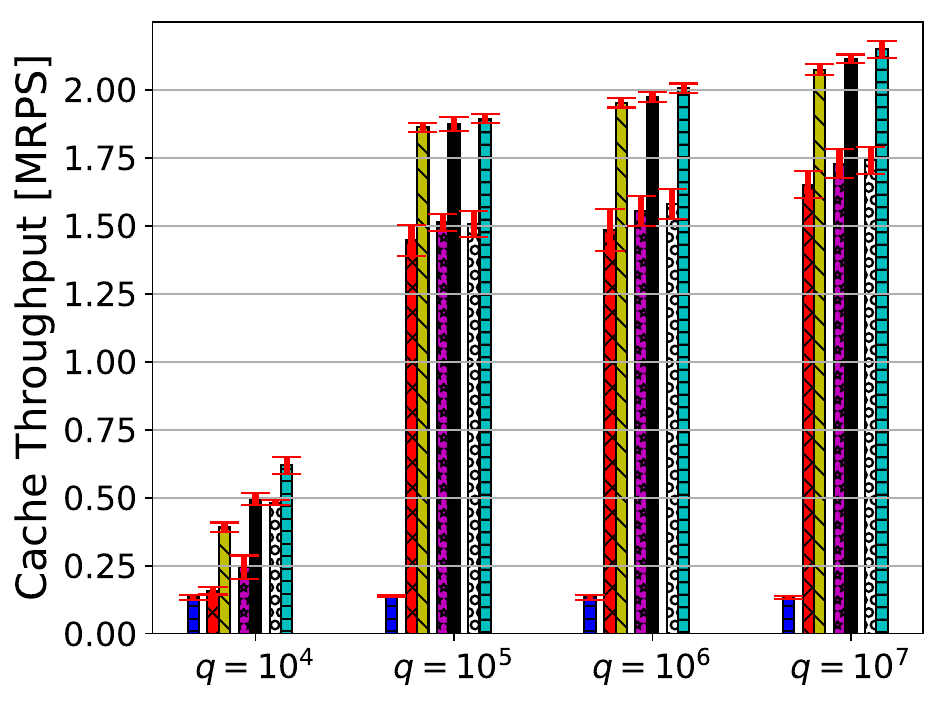}}\\ \vspace*{-3mm}
        \centering
        \subfloat[P1]
        {\includegraphics[width =0.3\columnwidth]
        {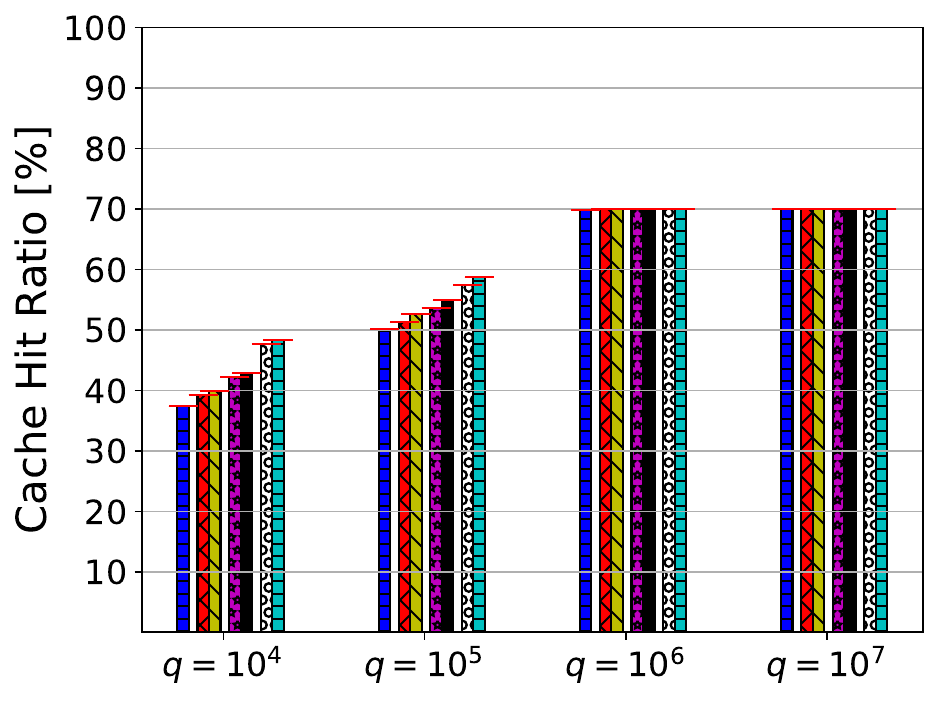}}
        \subfloat[F1]
        {\includegraphics[width =0.3\columnwidth]
        {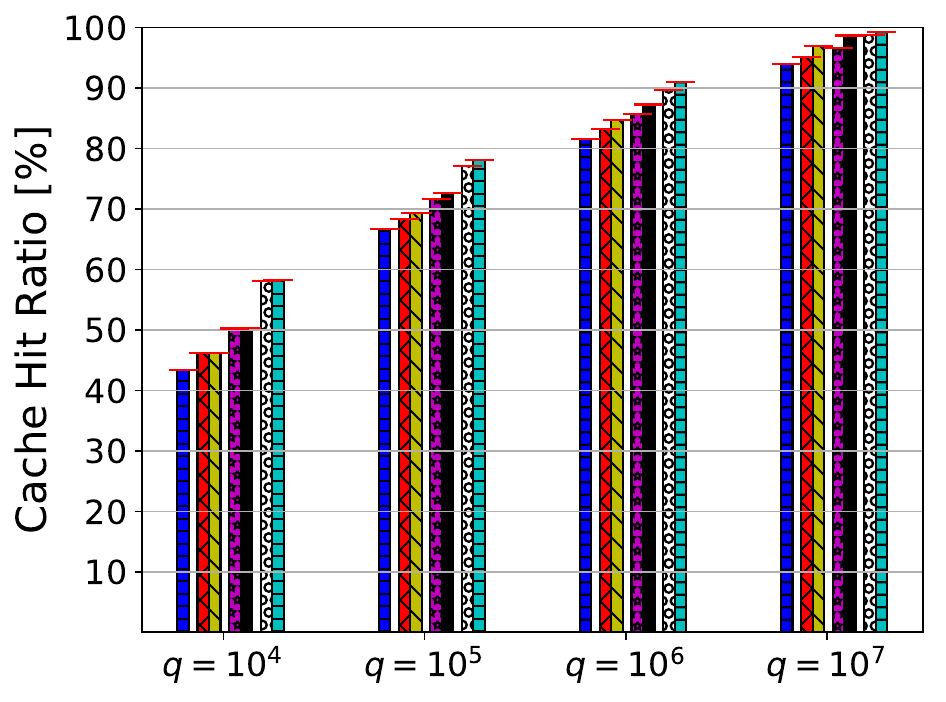}}
        \subfloat[WS1]
        {\includegraphics[width =0.3\columnwidth]
        {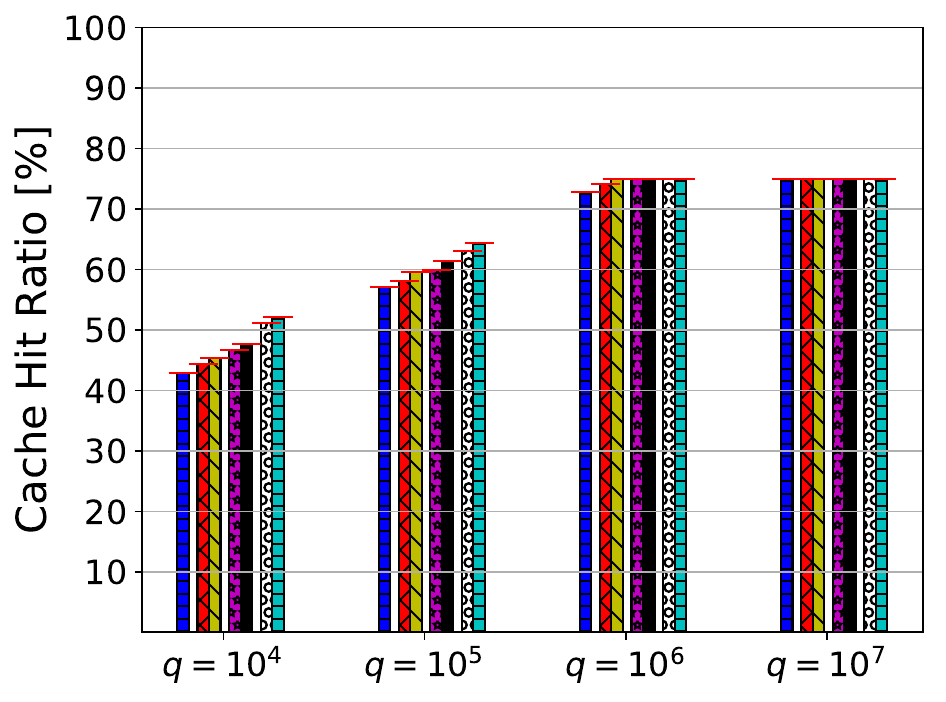}}\\ \vspace*{-1mm}
        \centering
        {\includegraphics[width =1.0010\columnwidth]
        {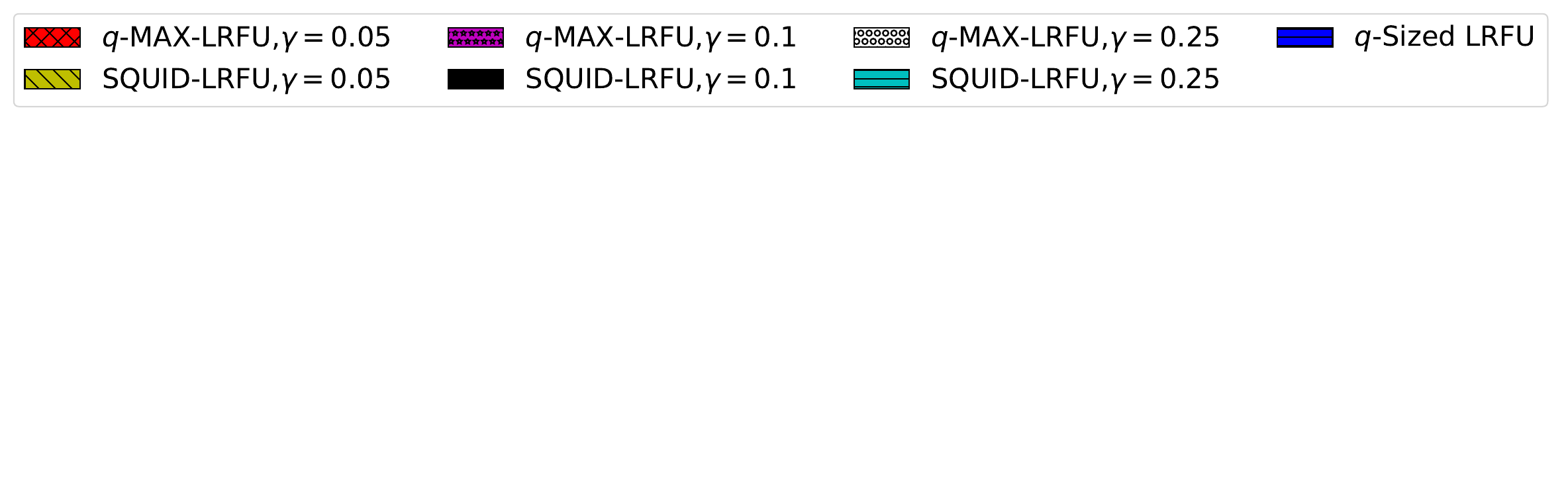}} \vspace{-4cm}
        \caption{Throughput (Million Requests Per Second) and Hit ratio of LRFU cache ($c = 0.75$) implemented using \qMAX, \sys and Heap implementation of LRFU (q-sized LRFU) on the caching datasets. \label{fig:lrfu_Hit_Ratio}}
        \vspace{-4mm}
      \end{figure}

\Wenchen{Figure out what to put in the "Application section.}


\begin{figure*}[]
    \vspace{-2mm}
    \subfloat[NWHH $q{=}10^{6}$]
    {\includegraphics[width =0.31\columnwidth]{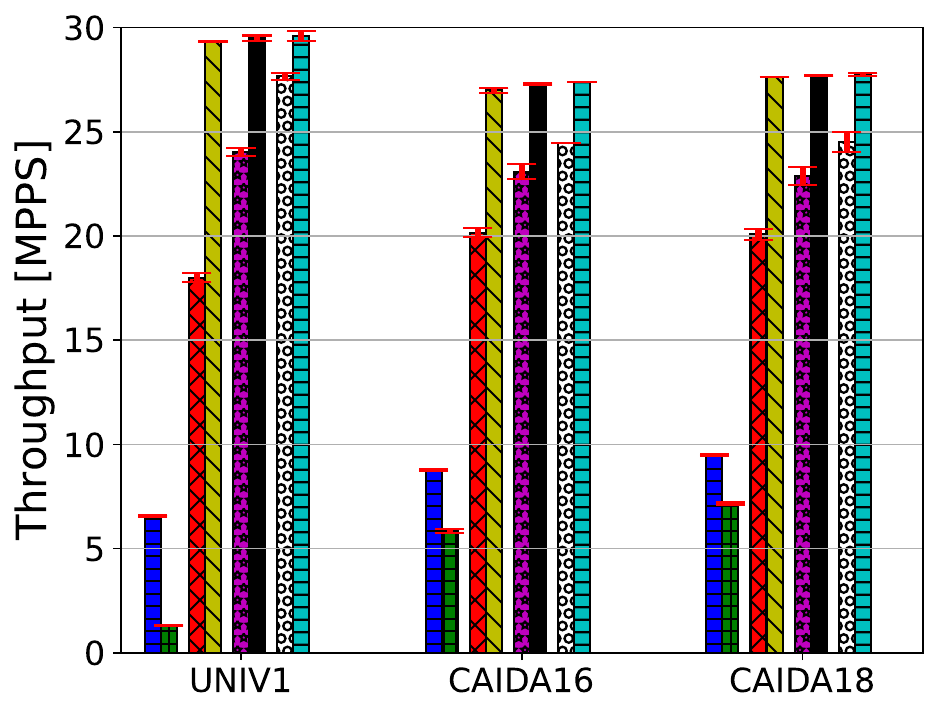}}
    \subfloat[PS $q=10^{6}$]
    {\includegraphics[width =0.31\columnwidth]
    {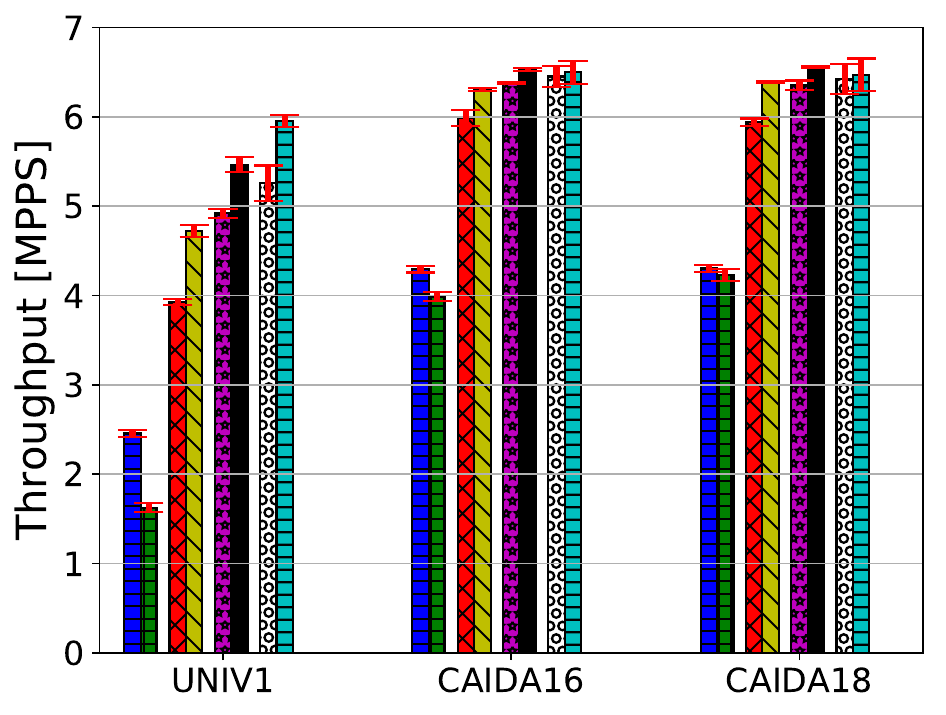}}
    \subfloat[PBA $q=10^{6}$]
    {\includegraphics[width =0.31\columnwidth]
    {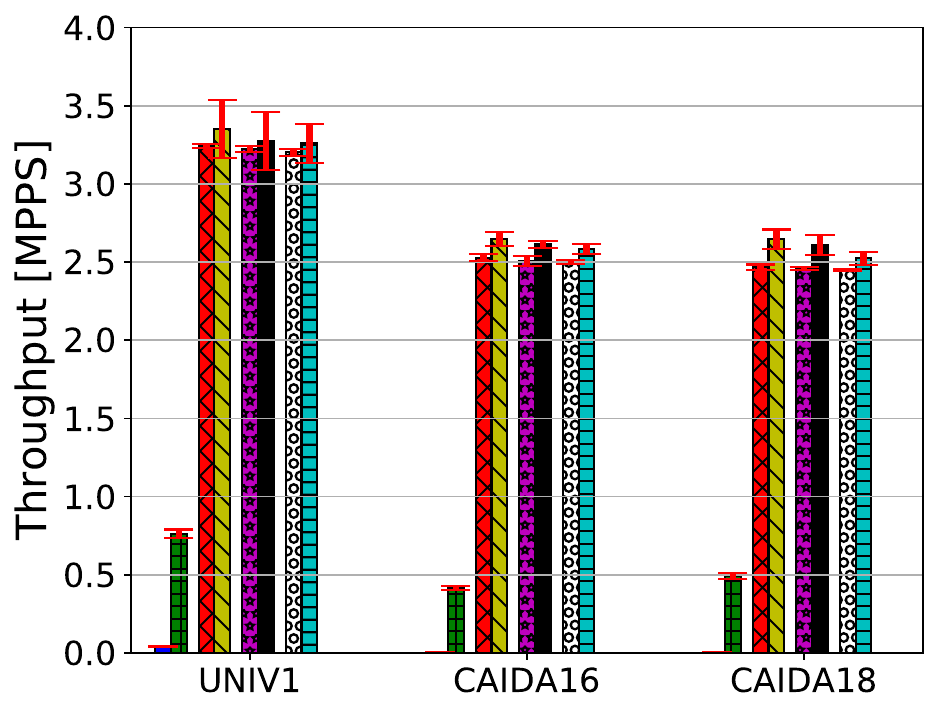}}
    
    \centering
    {\includegraphics[width =.6\columnwidth]
    {newGraphs/legend.pdf}}
    \vspace{-30mm}
    \caption{Throughput of various applications when implemented using \qMAX and \sys.}   \label{fig:qmax_application_apndx}
    \vspace{-5mm}
  \end{figure*}  
           
\vspace{-0.2cm}
\section{CPU LRFU Evaluation}\label{app:LRFU}
In this section, we implement an LRFU~\cite{LRFU} based on \sys{}-HH. Namely, we use the Cuckoo hash table to store the cached items and implicitly delete items whose scores are below the water level. This way, a logically deleted item can still yield a cache hit until it is overwritten by a different entry.
LRFU assigns a score that combines recency and frequency to each cached entry and retains the highest-ranked entries at all times. It is known to achieve high hit ratios but requires a logarithmic runtime~\cite{ARC}. 
The logarithmic runtime is explained by the existing implementations keeping cached items ordered according to the score. Here, \sys{}-HH is used to maintain the highest score items, whereas we use the water level technique to remove low-score items lazily. Thus, our algorithm works at a constant complexity. 

We used \sys-HH and \qMAX{} to implement LRFU, as well as standard heap and Skiplist-based implementations. Figure~\ref{fig:lrfu_Hit_Ratio}(a)-(c) shows that implementing LRFU with \sys-HH has the throughput across many cache sizes and traces. Here, the improvement is due to the water level technique of \sys{}-HH that reduces the maintenance overheads over existing implementations~\cite{qMax}. 

To complete the picture, Figure~\ref{fig:lrfu_Hit_Ratio}(d)-(f) shows that the hit ratio of all LRFU implementations is similar to that of LRFU in all datasets and cache sizes. The differences in hit-ratio vary with $\gamma$ as our cache contains between $q$ and $(1+\gamma)\cdot q$ items. The keen observer can also notice that \sys{}-HH offers a slight benefit to hit-ratio as it retains slightly more items in the cache due to only removing items upon an insert, whereas \qMAX{} removes items in batches.

\begin{figure*}[]
    \subfloat[MergeP (Caching)]
    {\includegraphics[width =0.32\columnwidth]
    {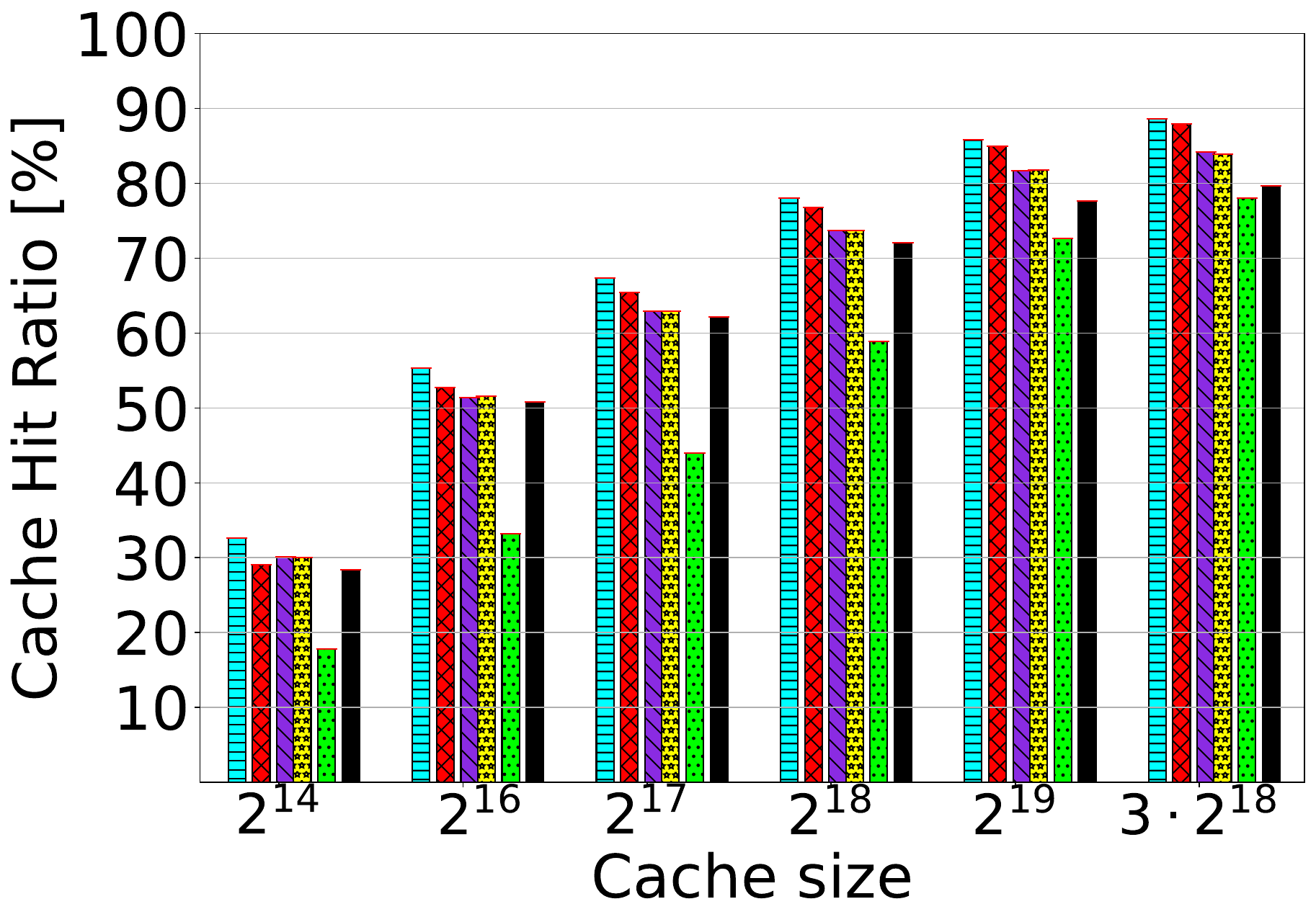}} 
    \subfloat[Zipf $0.9$]
    {\includegraphics[width =0.32\columnwidth]
    {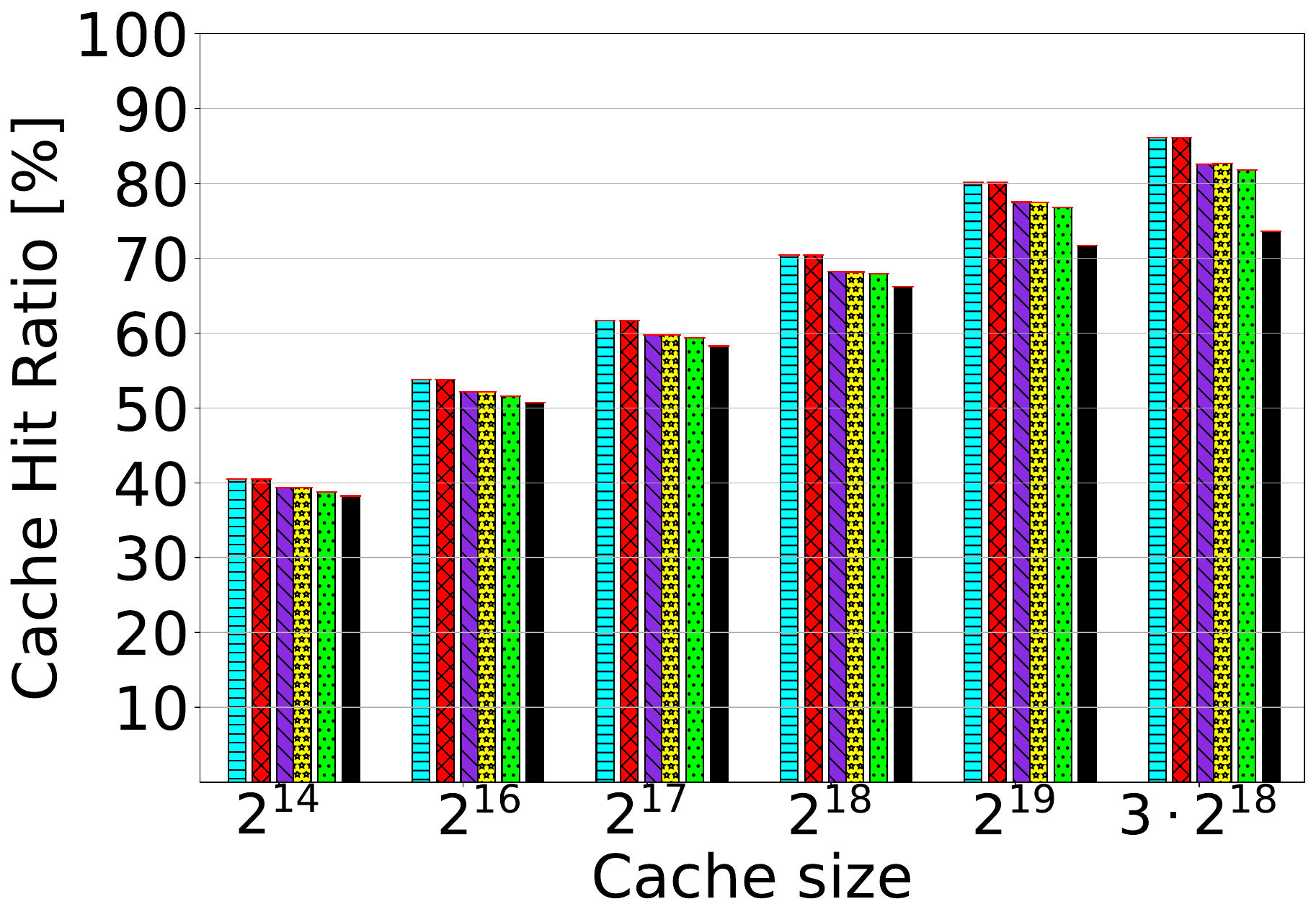}} 
    \subfloat[Zipf $0.99$]
    {\includegraphics[width =0.32\columnwidth]
    {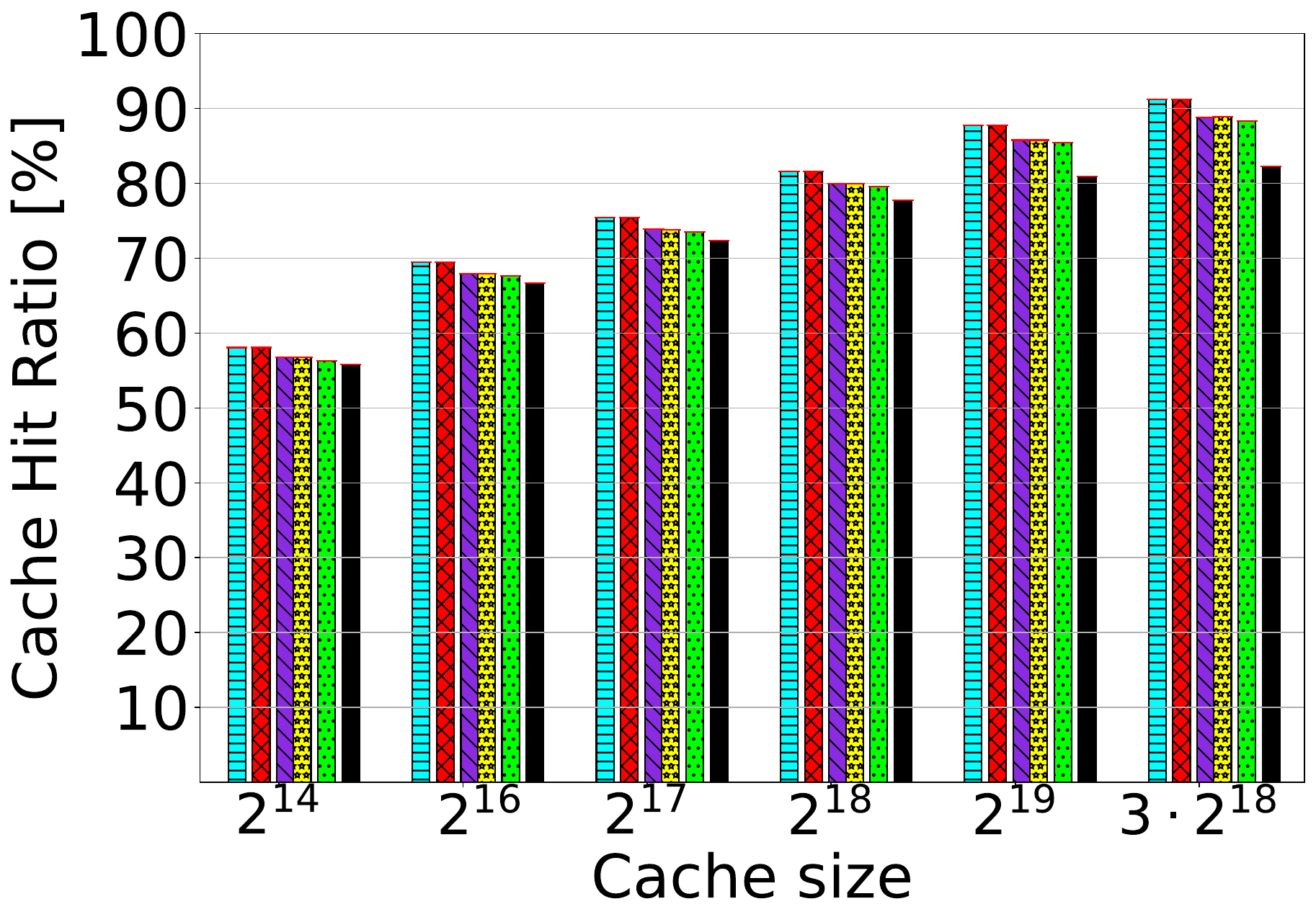}} \\
    \centering
    {\includegraphics[width =.6632\columnwidth]
    {newGraphs/legend_cache.pdf}}\vspace{-0.2cm}
    \caption{\revise{Hit ratios of our P4-based \sys on Zipf 0.9, Zipf 0.99 and MergeP datasets.\Wenchen{To move to the "full version".}\vspace{-0.1cm}}}
    \label{fig:p4-lrfu_Hit_Ratio-other}
\end{figure*}

\begin{figure*}[]
    \centering
    {\includegraphics[width =.1632\columnwidth]
    {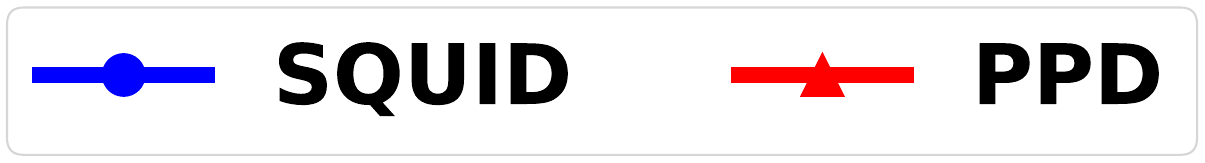}}\vspace{-0.3cm} \\
    \subfloat[CAIDA 2018]
    {\includegraphics[width =0.3\columnwidth]
    {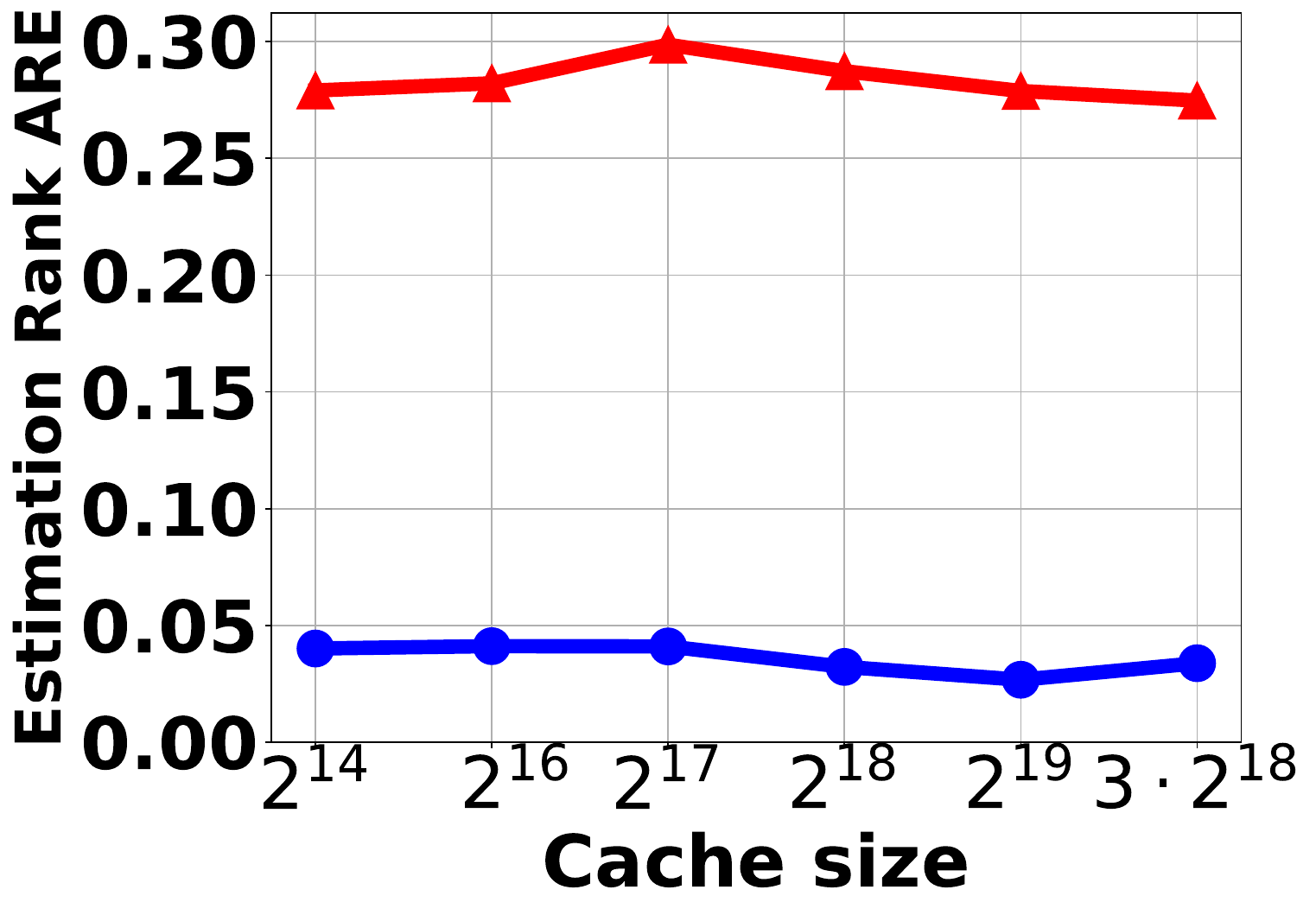}}
    \subfloat[S2 (Caching)]
    {\includegraphics[width =0.3\columnwidth]
    {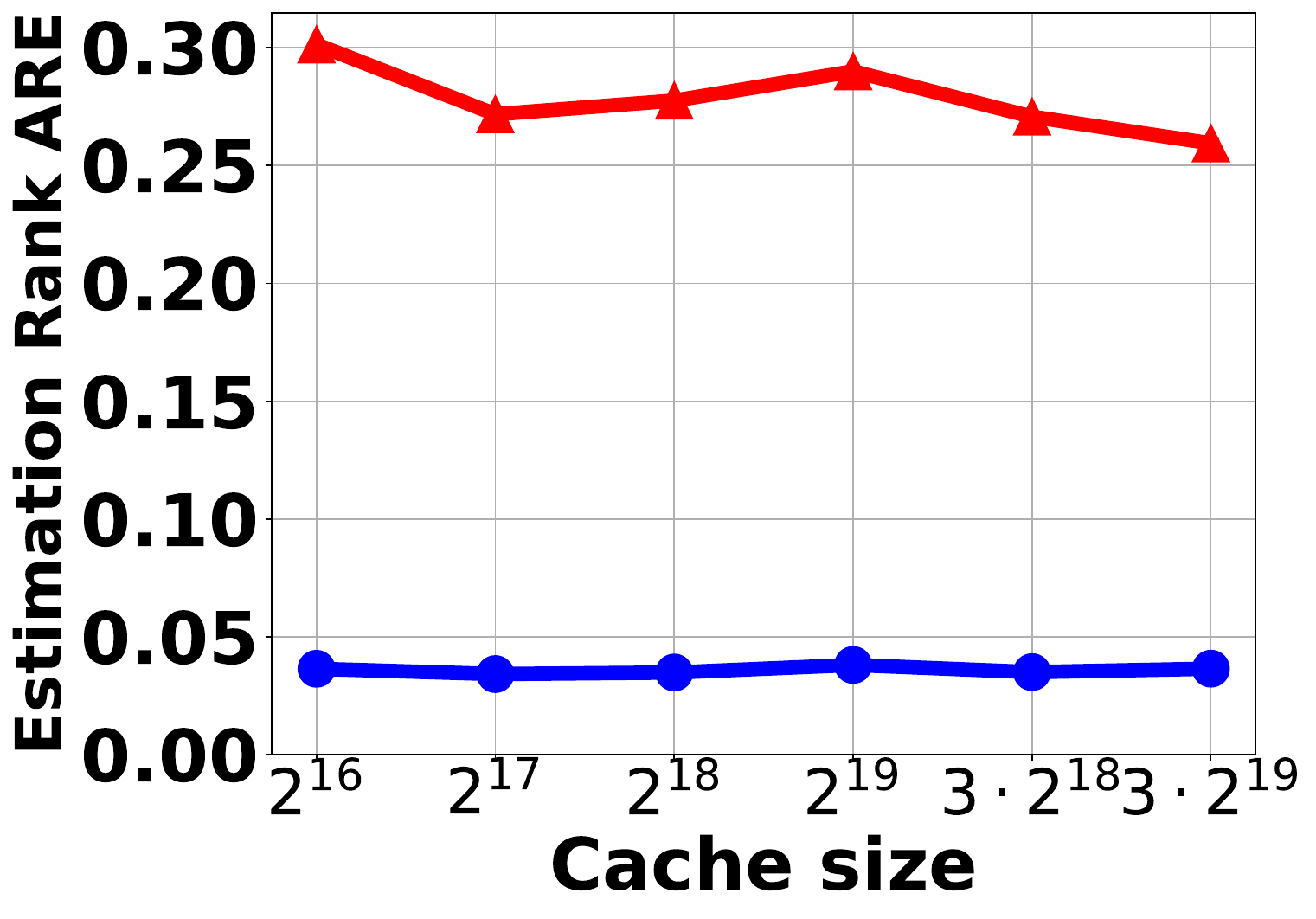}}
     \subfloat[MergeP (Caching)]
    {\includegraphics[width =0.3\columnwidth]
    {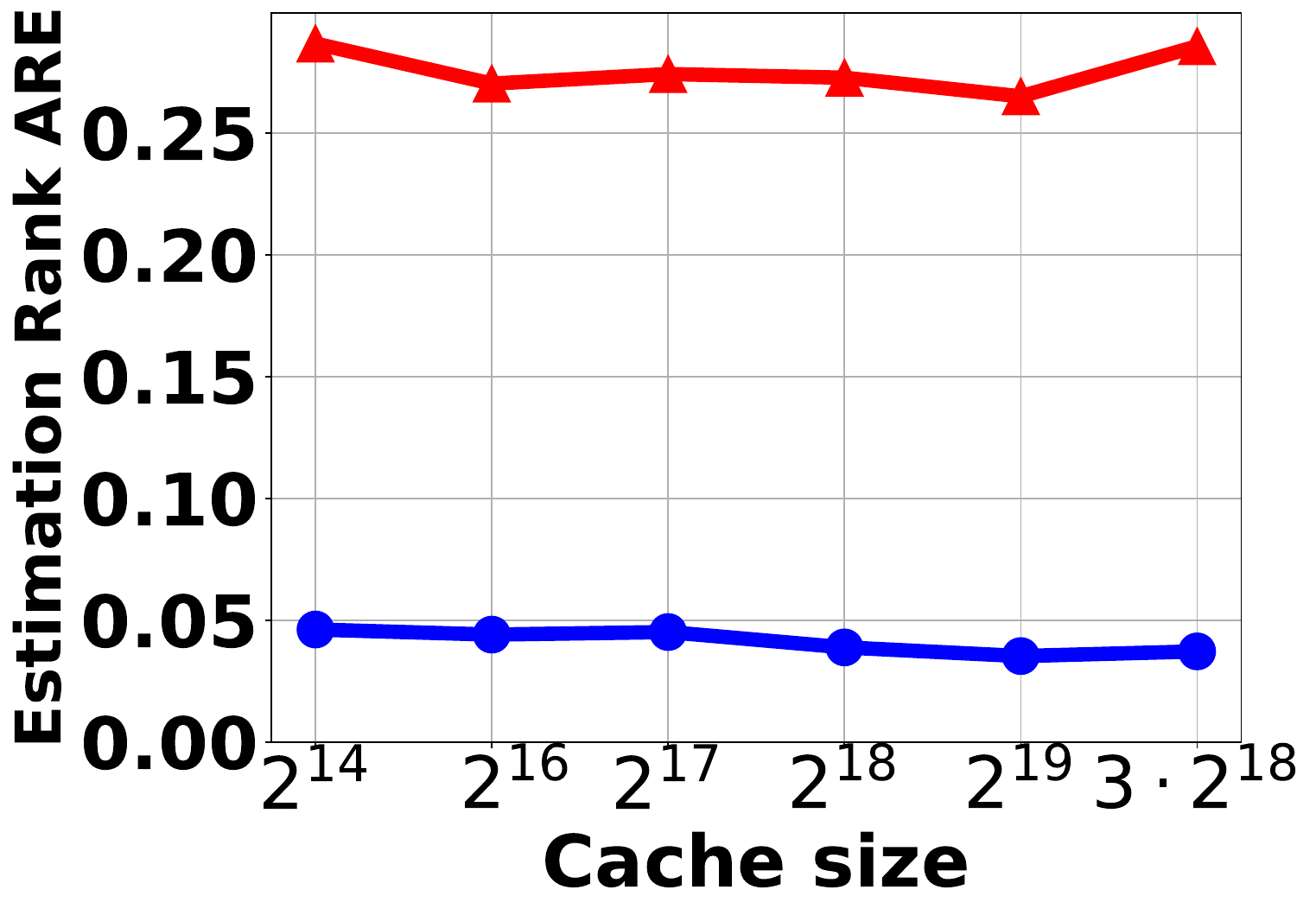}}
    
    
    \subfloat[Zipf $0.9$]
    {\includegraphics[width =0.3\columnwidth]
    {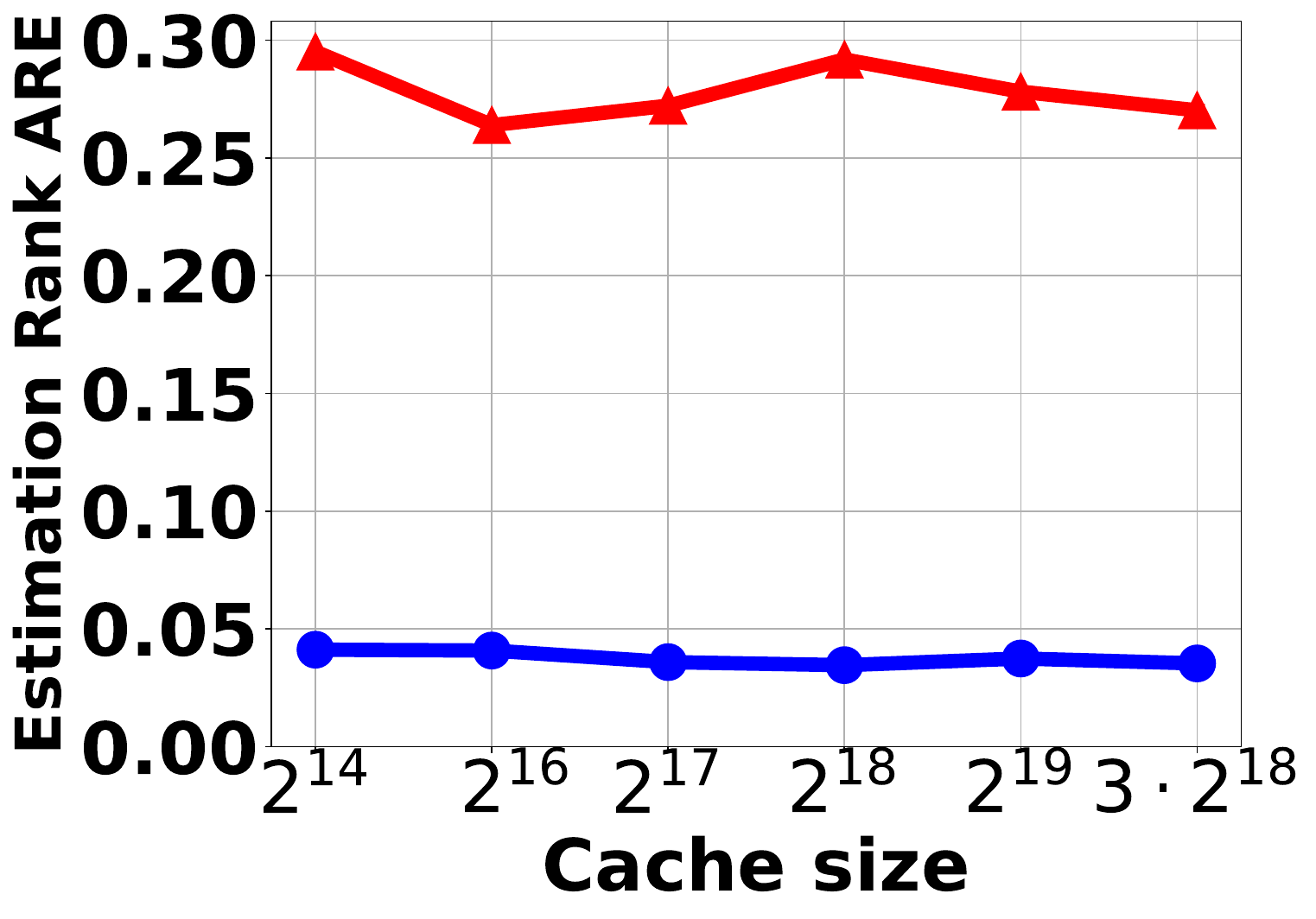}}
    \subfloat[Zipf $0.95$]
    {\includegraphics[width =0.3\columnwidth]
    {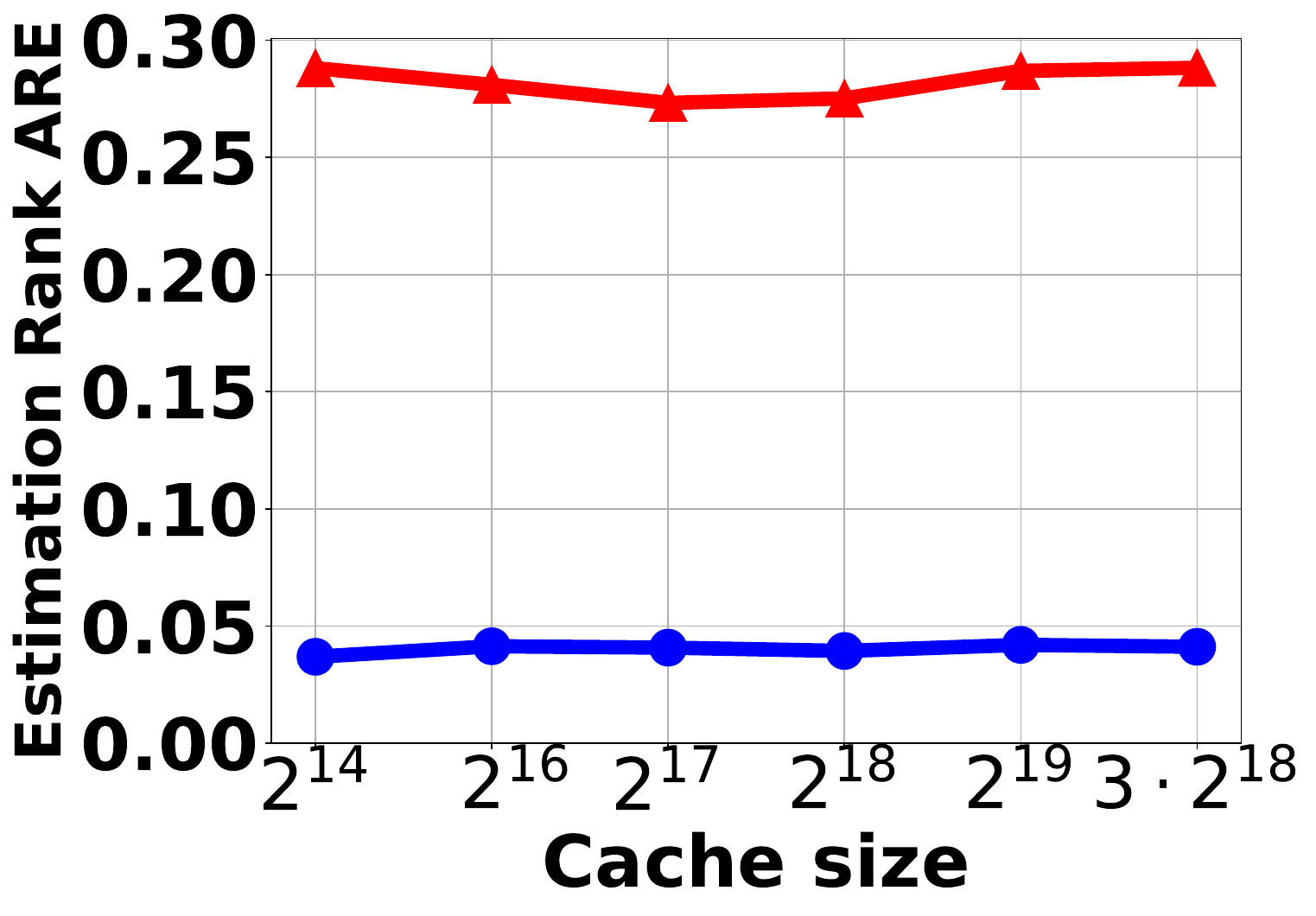}} 
    \subfloat[Zipf $0.99$]
    {\includegraphics[width =0.3\columnwidth]
    {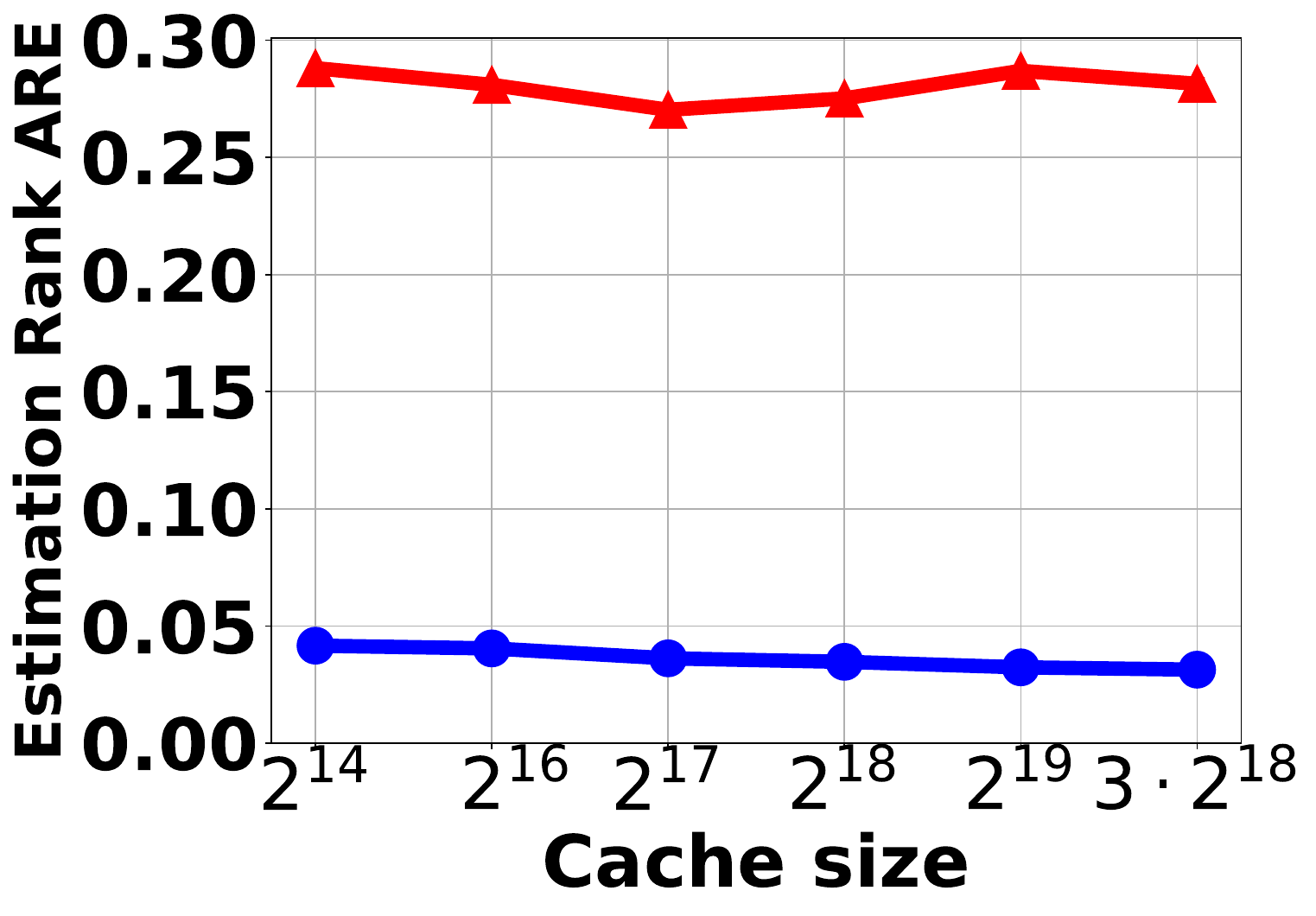}} 
    
    
    \vspace{-0.4cm}
    \caption{\revise{Relative error of quantile estimation (Section~\ref{subsubsec:settings}) comparing \sys with PPD~\cite{ppd}.}}\vspace{-0.2cm} 
    \label{fig:p4-quantile-estimation}
\end{figure*}

\end{document}